\documentclass[11pt]{article}

\usepackage{mathpazo}
\usepackage{times}

\usepackage{amsthm}
\usepackage{amsmath}  
\usepackage{amsfonts}
\usepackage{amssymb}  
\usepackage{mathrsfs} 
\usepackage{mathtools}
\usepackage{cite}     

\usepackage{tikz}

\usepackage{fullpage}

\usepackage{xfrac}
\usepackage{epstopdf}
\usepackage{verbatim}

\title{\Huge $\,$\\[-6.00ex]
Binary Linear Codes with Optimal Scaling:\\
Polar Codes with Large Kernels\\[0.90ex]}

\author{
Arman Fazeli\\
   \small University of California San Diego\vspace*{-0.72ex}\\
   \small 9500 Gilman Drive, La Jolla, CA\,92093\vspace*{-0.54ex}\\
   \ttfamily\bfseries\small afazelic@ucsd.edu\\[4.5ex]
\and
{Hamed Hassani}\\
   \small University of Pennsylvania\vspace*{-.72ex}\\
   \small 3101 Walnut St, Philadelphia, PA\,19104\vspace*{-0.54ex}\\
   \ttfamily\bfseries\small hassani@seas.upenn.edu\\[4.5ex]
\and
{Marco Mondelli}\\
   \small Institute of Science and Technology Austria (IST Austria)\vspace*{-.72ex}\\
   \small Am Campus 1, 3400 Klosterneuburg, Austria\vspace*{-0.54ex}\\
   \ttfamily\bfseries\small marco.mondelli@ist.ac.at\\[4.5ex]
\and
{Alexander Vardy}\\
   \small University of California San Diego\vspace*{-0.72ex}\\
   \small 9500 Gilman Drive, La Jolla, CA\,92093\vspace*{-0.54ex}\\
   \ttfamily\bfseries\small avardy@ucsd.edu\\[6.5ex]

\thanks{%
The research of 
Arman Fazeli and Alexander Vardy was supported in part 
by the United States National Science Foundation under 
Grants CCF-1405119 and CCF-1719139.
Marco Mondelli was supported in part
by an Early Postdoc Mobility Fellowship from the
Swiss National Science Foundation.\vspace*{-6.00ex}}
}

\newtheoremstyle{custom}
{} 
{} 
{} 
{} 
{\bfseries} 
{:} 
{.25em} 
{} 
\theoremstyle{plain}

\newtheorem{theorem}{Theorem}
\newtheorem{lemma}[theorem]{Lemma}
\newtheorem{proposition}{Proposition}
\newtheorem{definition}{Definition}

\newtheorem*{theorem*}{Theorem}
\newtheorem*{definition*}{Definition}

\newcounter{enumrom}
\renewcommand{\theenumrom}{(\roman{enumrom})}

\makeatletter
\renewcommand{\@endtheorem}{\endtrivlist}
\makeatother

\makeatletter
\renewcommand{\fnum@figure}{{\bf Figure\,\@arabic\c@figure}}
\makeatother

\newcommand{\cA}{{\cal A}} 
\newcommand{\cB}{{\cal B}}

\newcommand{\cF}{{\cal F}}

\DeclareMathAlphabet{\mathbfsl}{OT1}{ppl}{b}{it} 

\newcommand{\bU}{\mathbfsl{U}} 
\newcommand{\bV}{\mathbfsl{V}}
 
\newcommand{\bX}{\mathbfsl{X}}

\newcommand{\uuu}{\mathbfsl{u}} 
\newcommand{\vvv}{\mathbfsl{v}}
 
\newcommand{\xxx}{\mathbfsl{x}}
\newcommand{\yyy}{\mathbfsl{y}}

\newcommand{\ceil}[1]{\left\lceil #1 \right\rceil}
\newcommand{\floor}[1]{\left\lfloor #1 \right\rfloor}

\newcommand{\be}[1]{\begin{equation}\label{#1}}
\newcommand{\ee}{\end{equation}} 
\newcommand{\eq}[1]{(\ref{#1})}

\renewcommand{\le}{\leqslant} 
\renewcommand{\leq}{\leqslant}
\renewcommand{\ge}{\geqslant} 
\renewcommand{\geq}{\geqslant}

\newcommand{\script}[1]{{\mathscr #1}}

\renewcommand{\Bbb}{\mathbb}
 
\newcommand{\N}{{\Bbb N}}
\newcommand{\R}{{\Bbb R}}

\newcommand{\Tref}[1]{The\-o\-rem\,\ref{#1}}
\newcommand{\Pref}[1]{Pro\-po\-si\-tion\,\ref{#1}}
\newcommand{\Lref}[1]{Lem\-ma\,\ref{#1}}
\newcommand{\Cref}[1]{Co\-ro\-lla\-ry\,\ref{#1}}

\newcommand{\Ftwo}{{{\Bbb F}}_{2}}

\newcommand{\Strut}[2]{\rule[-#2]{0cm}{#1}}

\newcommand{\al}{\alpha}

\newcommand{\eps}{\varepsilon}
\renewcommand{\epsilon}{\varepsilon}

\newcommand{\Arikan}{Ar{\i}kan}

\newcommand{\Km}{K^{\raisebox{0.5pt}{${\scriptscriptstyle\otimes}\scriptstyle m$}}}

\newcommand{\bits}{\{0,1\}}

\newcommand{\Pe}{P_{\kern-1pt\rm e}}

\makeatletter
\DeclareRobustCommand{\sbinom}{\genfrac[]\z@{}}
\makeatother
\newcommand{\G}[2]{\sbinom{{#1}}{{#2}}}

\DeclareMathOperator{\BEC}{BEC}
\DeclareMathOperator{\GL}{GL}

\renewcommand{\P}{{\mathbb P}}

\newcommand{\zero}{{\mathbf 0}}
\newcommand{\one}{{\mathbf 1}}

\newcommand{\sY}{\script{Y}}

\newcommand{\VAB}{\bV_{\hspace{-1pt}\mbox{\scriptsize$\cA \cup\kern0.3pt \cB$}}}

\makeatletter

\gdef\@punct{.\ \ }  
\def\@sect#1#2#3#4#5#6[#7]#8{%
  \ifnum #2>\c@secnumdepth
     \def\@svsec{}
  \else
     \refstepcounter{#1}\edef\@svsec{%
     \ifnum #2>0{{\csname the#1\endcsname}}.\fi%
    \hskip .5em}
  \fi
  \@tempskipa #5\relax
  \ifdim \@tempskipa>\z@
     \begingroup #6\relax
       \@hangfrom{\hskip #3\relax\@svsec}{\interlinepenalty \@M #8\par}
     \endgroup
     \csname #1mark\endcsname{#7}
     \addcontentsline{toc}{#1}{\ifnum #2>\c@secnumdepth\else
          \protect\numberline{\csname the#1\endcsname}\fi#7}
  \else
     \def\@svsechd{#6\hskip #3\@svsec #8\@punct\csname #1mark\endcsname{#7}
     \addcontentsline{toc}{#1}{\ifnum #2>\c@secnumdepth \else
          \protect\numberline{\csname the#1\endcsname}\fi#7}}
  \fi
  \@xsect{#5}}

\def\@ssect#1#2#3#4#5{\@tempskipa #3\relax
  \ifdim \@tempskipa>\z@
    \begingroup #4\@hangfrom{\hskip #1}{\interlinepenalty \@M #5\par}\endgroup
  \else \def\@svsechd{#4\hskip #1\relax #5\@punct}\fi
  \@xsect{#3}}

\makeatother

\begin{document}
\maketitle

\thispagestyle{empty}

\begin{abstract}
	\vspace*{0.72ex}
\color{black}{	
	\looseness=-1\noindent
	We prove that, for the binary erasure channel (BEC),
	the polar-coding paradigm gives rise to codes that not only approach the Shannon limit but do so under the \emph{best possible scaling of their block length} as a~function of the gap to capacity. This result exhibits the first known family of binary codes that attain both optimal scaling and quasi-linear complexity of encoding and decoding.
	%
	Our proof is based on the construction and analysis of binary polar codes 	\emph{with large kernels}. When communicating reliably at rates within 	$\varepsilon > 0$ of capacity, the code length $n$ often scales as $O(1/\varepsilon^{\mu})$, where the constant $\mu$ is called the \emph{scaling exponent}. It is known that the optimal scaling exponent is $\mu=2$, and it is achieved by random linear codes. The scaling exponent of conventional polar codes (based on the $2\times 2$ kernel) on the BEC is
	$\mu=3.63$. This falls far short of the optimal scaling guaranteed by random codes. Our main contribution is a rigorous proof of the following result: for the BEC, 	there exist $\ell\times\ell$ binary kernels, such that polar codes constructed from these kernels achieve scaling exponent $\mu(\ell)$ that tends to the optimal value of $2$ as $\ell$ grows. We furthermore characterize precisely how large $\ell$ needs to be  as a  function of the gap between $\mu(\ell)$ and $2$. The resulting binary codes maintain the recursive  structure of conventional polar codes, and thereby achieve construction complexity $O(n)$ and encoding/decoding complexity $O(n\log n)$.
	\vspace*{6.00ex}
}
\end{abstract}


\newpage
\setcounter{page}{1}
$\,$\vspace*{-6.0ex}
\section{Introduction} 
\label{sec:Introduction}
\vspace{-1.00ex}

\noindent\looseness=-1
Shannon's 
coding theorem implies that for every
binary-input memoryless symmetric (BMS) channel $W$, there
is a capacity $I(W)$ 
such that the following holds:
for all $\eps > 0$ and $\Pe > 0$, there exists~a~bina\-ry code
of rate at least $I(W) - \eps$ which
enables communication over $W$ with probability of error
at most $\Pe$. Ever since the publication of Shannon's
famous paper~\cite{Shannon}, the holy grail of coding 
theory was to find explicit codes that achieve Shannon
capacity with polynomial-time complexity of construction
and decoding. 
Today, several such families of codes are known, and the
principal remaining challenge is to characterize 
\emph{how fast we can approach capacity} as a~function
of the code block length $n$. Specifically, we now have
explicit binary codes
(which can be constructed and decoded in polynomial time)
of length $n$ and rate $R$, such that the gap to capacity
$\epsilon = I(W)-R$ required to achieve any fixed
error probability $\Pe > 0$ vanishes as a function of~$n$. 
The fundamental theoretical problem is to characterize
how fast this happens. Equivalently, we can fix $\epsilon = I(W)-R$
and ask how large does the block length $n$ need to be 
as a function of $\epsilon$. That is, we are interested 
in the \emph{scaling between the block length and the gap to capacity},
under the constraint of polynomial-time construction
and decoding.

\textcolor{black}{Based on~\cite{PPV10},} it is known that the optimal scaling is of the form $n = O(1/\varepsilon^{\mu})$, where the constant $\mu$ is \linebreak referred to as the \emph{scaling~exponent}. 
It is furthermore known that the best possible
scaling exponent is $\mu=2$, and it is achieved by
random linear codes --- although, of course, random
codes do not admit efficient decoding.
In this paper, we present the first family of binary codes 
that attain both optimal scaling and quasi-linear complexity
on~the binary erasure channel (BEC). Specifically,
for any fixed $\delta > 0$, we exhibit codes that
ensure reliable communication on the BEC at rates 
within $\varepsilon > 0$ of the 
Shannon capacity, with block length $n=O(1/\varepsilon^{2+\delta})$,
construction complexity $O(n)$, and encoding/decoding complexity
$O(n\log n)$. 

To establish this result, we use \emph{polar coding}, invented
by \Arikan~\cite{Ari09} in 2009.
However, while \Arikan's polar codes are based upon a specific
$2 \times 2$ kernel, we use $\ell \times \ell$ binary polarization
kernels, where $\ell$ is a sufficiently large constant.
The main technical challenge is to prove that this construction
works. To this end, we choose
the polarization kernel uniformly at random
from the set of all $\ell\times\ell$ nonsingular binary matrices,
and show that with probability at least $1 - O(1/\ell)$,
the resulting scaling exponent is at most $2 + \delta$. 
Since $\ell$ is a constant that depends only on $\delta$, 
the choice of a polarization
kernel can be, in principle, de-randomized using brute-force search
whose complexity is independent of the block length.

\textcolor{black}{By way of a disclaimer, the theorems in this paper require the size of the kernel to be \emph{extremely large} to the point that they are not practical at all.  Moreover, the decoding complexity of polar codes constructed with large kernels, if done naively and over arbitrary channels, scales \emph{exponentially} with the kernel size, which is another challenge in using these codes. On the positive side, we prove that, for sufficiently large values of $\ell$, \emph{almost all} kernels yield fast scaling exponents. This establishes a new family of binary error-correcting codes that are \emph{near optimal} in every aspect at the asymptotic regime; this is of purely theoretical interest. The problem of finding good polarization kernels with reasonable size and decoding complexity remains unsolved. We will review some recent advancements in the field on this problem in the final section.} 

\looseness=-1
The rest of this paper is organized as follows. 
In this section, we provide the necessary background and give an informal
statement of our main result (\Tref{thm1.2}).
In Section\,\ref{sec:outline}, we
present a brief outline of the proof. In\linebreak
Section\,\ref{sec:main_theorem}, 
we formally state our main theorem (\Tref{thm:main_theorem1}), and
gradually reduce its proof to a certain statement about 
$\ell \times \ell$ binary polarization kernels
(\Tref{thm:main_theorem2}). We defer the proof of \Tref{thm:main_theorem2}
itself, which is technically quite elaborate, 
to Section\,\ref{appendix:proof}. 
We conclude with a brief discussion in Section\,\ref{sec:disc}.

\subsection{Background and context}\vspace{-0.54ex}
\label{subsec:over}

A sequence of papers, starting with \cite{Do61,St62} in 
1960s and culminating with \cite{Ha09,PPV10},
shows that for any discrete memoryless
channel $W$ and \emph{any} code of length $n$ and rate $R$ that
achieves error-probability $\Pe$ on $W$, we have
\be{eq:randomsc}
I(W) - R 
\ \ge \
\frac{\text{const}(\Pe,W)}{\sqrt{n}} \ - \ O\left(\frac{\log n}{n}\right),
\ee
where the constant (which is given explicitly in~\cite{PPV10}) 
depends on $W$ and $\Pe$, but not on $n$. This immediately 
implies that if $n = O\left(1/\epsilon^\mu\right)$,
where $\eps = I(W) - R$ is the gap to capacity, then
$\mu \ge 2$. We further~note~that expressions similar
to \eq{eq:randomsc} were derived from the perspective 
of threshold phenomena in \cite{Tillich-cpc00} and 
from the~perspective of statistical physics in \cite{Mon01b}. 
%
The fact that $\mu \ge 2$ also follows from 
a simple heuristic~argument. For simplicity, consider the
special case of transmission over the BEC 
with erasure probability $p$. As $n\to\infty$,
the number of erasures will tend to the normal distribution
with mean $np$ and standard deviation \smash{$\sqrt{n p(1-p)}$}. 
Thus, channel randomness yields a variation in the fraction of erasures 
of order $1/\sqrt{n}$. This indicates that, in order to achieve a fixed 
error probability, the gap to capacity $\epsilon$ has to scale at least as
$1/\sqrt{n}$.

It is well known~\cite{Ha09,PPV10} that the lower bound $\mu = 2$ 
is achieved by random linear codes. For the special case of
transmission over the BEC, the proof of this fact reduces to
computing the rank of a certain random matrix. Indeed, the 
generator matrix of a random linear code of length $n$ and rate $R$
is a matrix with $Rn$ rows and~$n$ columns whose entries 
are i.i.d.\ uniform in $\{0, 1\}$. The effect of transmission 
over the BEC with erasure probability $p$ is equivalent to
removing each column of this generator matrix independently with
probability~$p$. The probability of error (under maximum-likelihood
decoding) is thus equal to the probability that such residual matrix 
is not full-rank. This probability is easy to compute, and
the desired scaling result immediately follows.

\looseness=-1
Unfortunately, random linear codes cannot be decoded efficiently.
On general BMS channels, this task~is~NP-hard~\cite{BMvT}.
On the BEC, decoding a general binary linear code takes time $O(n^\omega)$,
where $\omega$ is the exponent of matrix multiplication.
This leads to the following natural question: what is the lowest
possible scaling expo\-nent for binary codes that can be constructed, encoded,
and decoded efficiently? For the BEC, we take \emph{efficiently}
to mean linear or quasi-linear complexity. 
Here is a brief survey of the current state of knowledge on this question.

Forney's concatenated codes~\cite{For65} are a classical example 
of a capacity-achieving family of codes. However, their
construction and decoding complexity are exponential in the
inverse gap to capacity $1/\epsilon$ (see~\cite{Guru-Xia, Guru-Xia2}
for more details), so they are \textcolor{black}{not competitive from an asymptotic perspective}.
In recent years, three new families of capacity-achieving
codes have been discovered; let us review what is known regarding
their scaling exponents.\vspace{-1.00ex}

\begin{description}
\item[Polar codes:] 
\!\!Achieve the capacity of any BMS channel under a 
successive-cancellation decoding algorithm~\cite{Ari09} 
that runs in time $O(n \log n)$. 
It was shown in \cite{Guru-Xia, Guru-Xia2} that the block length,
construction complexity, and decoding complexity are all bounded by
a polynomial in 
$1/\epsilon$, which implies that the scaling exponent $\mu$ is finite. 
Later, a sequence of papers \cite{KMTU, HAU, Goldin-Burshtein, MHU2}
provided rigorous upper and lower bounds on $\mu$.
The specific value of $\mu$ depends on the channel $W$.
It is known that $\mu=3.63$ on the BEC. The best-known
bounds valid for any BMS channel $W$ are given by $3.579\le \mu \le 4.714$.
\vspace{-0.50ex}
   
\item[Spatially-coupled LDPC codes:]
\!\!Achieve the capacity of any BMS channel under a 
belief-propagation decoding algorithm \cite{KRU13}
that runs in linear time. A simple heuristic
argument yields that the scaling exponent of these 
codes is roughly $3$ (see \cite[Section VI-D]{MHU14asymm-ieeeit}). However,
a rigorous proof of this statement remains elusive and
appears to be technically challenging.
\vspace{-0.50ex}
   
\item[Reed-Muller codes:]
\!\! Achieve capacity of the BEC under maximum-likelihood
decoding \cite{RMpaper-STOC, RMpaper-ITTran} that runs in
time $O(n^\omega)$. 
While it has been observed empirically that the performance 
of Reed-Muller codes on the BEC is close to that of random 
codes~\cite{MoHaUr14}, no bounds on the
scaling exponent of these codes are known.
\end{description}

\looseness=-1
Let us point out that some papers also define a ``scaling exponent''
for codes that do not achieve capacity, such as ensembles of LDPC codes,
by substituting the specific threshold of the ensemble for 
channel capacity.
In this context, it is known~\cite{AMRU09} that 
for a large class of ensembles of LDPC codes and channel models, 
the scaling exponent is $\mu=2$. However, the threshold
of such LDPC ensembles does not converge to capacity.

\vspace{4.00ex}
\subsection{Our main result: 
    Binary linear codes with optimal scaling and quasi-linear complexity}
\vspace{-1.00ex}
\label{subsec:contribution}

Our main result provides the first family of binary codes 
for transmission over the BEC that
achieves optimal scaling between the gap to capacity $\epsilon$ and 
the block length $n$, and that can be constructed, encoded, 
and decoded in quasi-linear time. In other words, the block length, 
construction, encoding, and decoding
complexity are all bounded by a polynomial in $1/\epsilon$ and, moreover,
the degree of this polynomial approaches the information-theoretic lower
bound $\mu \ge 2$. Somewhat informally (cf.~\Tref{thm:main_theorem1}),
this result can be stated as follows.

\begin{theorem}
\label{thm1.2}
Consider transmission over \textcolor{black}{i.i.d. copies of} a binary erasure channel\, $W$ with capacity
$I(W)$. Fix \textcolor{black}{the block error probability $\Pe \in\! (0,1)$}~and an arbitrary $\delta \in (0, 1]$. Then, \textcolor{black}{there exists a fixed constant $\ell_0(\delta)$} such that for all
$R < I(W)$, there exists \textcolor{black}{a binary linear code} of rate \textcolor{black}{at least} $R$ that guarantees error probability at most $\Pe$ on the channel $W$, and whose block length $n$ \textcolor{black}{is at most}
\begin{equation}
n~ \ \textcolor{black}{\le \ \frac{\beta\ell_0(\delta)}{\bigl(I(W)-R\bigr)^{2+\delta}}}\hspace{1ex},
\end{equation}
where\/ $\beta = \bigl(1+2\hspace{0.1em}\Pe^{\textcolor{black}{-1}}\bigr)^{3}$
is a universal constant. 
\textcolor{black}{Furthermore, as $R$ approaches $I(W)$ and $n$ grows, this code has construction complexity\/ $O(n)$ and encoding/decoding complexity\, $O(n\log n)$. }
\end{theorem}

A \textcolor{black}{few} remarks regarding  \Tref{thm1.2}
are in order. First, in the definition of the
constant $\beta$, the term $\Pe$ is raised to the power of
$-1$. We point out that we could have similarly chosen any other negative constant as the exponent of~$\Pe$. \textcolor{black}{However, picking a smaller exponent for $\Pe$ requires us to select an upper bound on $\delta$ which is stricter than $\delta\le 1$. This, in turn, increases $\ell_0(\delta)$.} Second, the error
probability in \Tref{thm1.2} is upper bounded
by a fixed constant~$\Pe$. However, a somewhat stronger claim is possible.
It can be shown that \Tref{thm1.2} still holds if 
the error probability is required to 
decay \emph{polynomially fast} with the block length $n$.
\textcolor{black}{Lastly, it should be emphasized that $\ell_0(\delta)$ is a  constant that only depends on $\delta$. However, its dependence is of an exponential nature, \emph{i.e.} $\ell_0(\delta) = O(\exp(\sfrac{1}{\delta^{1.01}}))$. This limitation prevents the proposed scaling exponent to be exactly equal to $2$ while maintaining the quasi-linear complexity property.} 
\vspace{1.25ex}

To prove \Tref{thm1.2}, we will show that
there exist $\ell\times\ell$ binary kernels, such that polar codes
constructed~from these kernels achieve capacity with 
a scaling exponent $\mu(\ell)$ that tends to the optimal value of $2$ 
as $\ell$ grows. 
The claim regarding the construction 
and encoding/decoding complexities immediately~follows from known
results on polar codes~\cite{Ari09,Sasoglu-FnT,TVbounds}. 
Indeed, polar codes constructed from $\ell\times\ell$ binary kernels
maintain the recursive structure~of conventional polar codes, 
and thereby inherit construction complexity $O(n)$ and encoding/decoding
complexity $O(n\log n)$. We will discuss the decoding complexity 
in more detail in Section\,\ref{sec:disc}.

\vspace{2.00ex}
\subsection{A primer on polar codes}
\vspace{-1.00ex}
\label{subsec:primer}

Like many fundamental discoveries, polar 
codes~are rooted in a simple and beautiful basic idea.
Polarization is induced via a simple linear transformation
consisting of many Kronecker products of a binary
matrix~$K$, called the \emph{polarization kernel}, with itself.
Conventional polar codes, introduced by \Arikan~in~\cite{Ari09},
correspond to 
\begin{equation}\label{eq:orker}
K
\, = \,
\left[ 
\begin{array}{@{\hspace{0.50ex}}c@{\hspace{1.005ex}}c@{\hspace{0.50ex}}}
1 & 0\\
1 & 1\\ 
\end{array}
\right].
\end{equation}
However, it was shown in ~\cite{KSU} that
we can construct polar codes from any kernel $K$ 
that is an $\ell\times\ell$ nonsingular binary matrix, which
cannot be transformed into an upper triangular matrix under any column
permutations.

Let $W {:}\, \{0,1\} \to \sY$ be a 
BMS channel, characterized~in terms of 
its transition pro\-babilities $W(y|x)$,
for all $y \,{\in}\, \sY$ and $x \,{\in}\, \bits$.
Further, let $\bU\hspace{-1pt} = (U_1,U_2,\dots,U_n)$ be a block of 
\smash{$n = \ell^m$} bits chosen uniformly at random from 
\smash{$\{0,1\}^n$}. 
We encode $\bU$ as\, $\bX = \bU \Km$ and transmit $\bX$ through
$n$ independent copies of $W$, as shown in Figure~\ref{polar-def}.

\begin{figure}[h!] 
	\centering 
	\includegraphics[width=3.06in,trim=0 2ex 0 2ex]{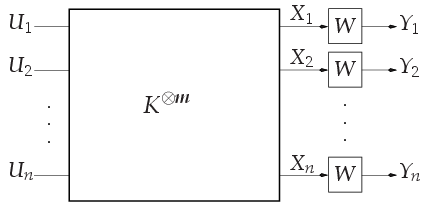}
	\caption{Block diagram of a polar coded communication scheme.}
	\label{polar-def}
\end{figure} 

To understand what polarization means~in this context, 
we consider a number of channels 
associated with this transformation (see also Chapter 5 of \cite{Sasoglu-FnT} and Chapter 2.4 of \cite{Hassani}). Let 
$W^n {:}\, \{0,1\}^n \to \sY^n$ be the channel
that corresponds to $n$ independent uses~of~$W$, and
let $W^* {:}\, \{0,1\}^n\! \to \sY^n$ be the channel
with transition probabilities given by
$
W^*\kern-1pt(\yyy|\uuu) 
=
W^n\kern-1pt
\bigl(\yyy\hspace{1pt}{\bigm|}\hspace{1pt}\uuu\hspace{1pt} \Km \bigr)
$. 
Finally, for all $i \in [n]\textcolor{black}{\triangleq \{1,2,\dots,n\}}$, let 
\smash{$W_i : \{0,1\}\! \to \sY^n{\times}\{0,1\}^{i-1}$}
be the channel that is ``seen'' by the bit $U_i$,
defined as\vspace*{0.9ex}
\be{Wi-def}
W_i\bigl( \yyy,\vvv | \hspace{1pt}u_i)
\:\ \textcolor{black}{\triangleq} \:\
\frac{1}{2^{n-1}}\hspace{-5pt}
\sum_{{\uuu}' \in \{0,1\}^{n-i}\hspace{-12pt}} \hspace{-6pt}
{W}^*\hspace{-1pt}\Bigl(\yyy\hspace{1pt}{\bigm|}\hspace{1pt}
(\vvv,u_i,{\uuu}') \Bigr)
\ = \
\frac{1}{2^{n-1}}\hspace{-5pt}
\sum_{{\uuu}' \in \{0,1\}^{n-i}\hspace{-12pt}} \hspace{-6pt}
{W}^n\hspace{-1pt}\Bigl(\yyy\hspace{1pt}{\bigm|}\hspace{1pt}
(\vvv,u_i,{\uuu}') \Km \Bigr),
\ee
where $(\cdot,\cdot)$ stands for 
concatenation. 
We say that $W_i$ is the \emph{$i$-th bit-channel}.
It is easy to see that \smash{$W_i\bigl( \yyy,\vvv | \hspace{1pt}u_i)$}~is 
indeed the~probability of the event that
$(Y_1,Y_2,\dots,Y_n) \hspace{-1pt}= \yyy\/$ and 
$(U_1,U_2,\dots,U_{i-1}) \hspace{-1pt}= \vvv$
given that $U_i = u_i$.

The key observation of \cite{Ari09} is that, as $n$ grows, the 
$n$ {bit-channels} $W_i$ defined in \eq {Wi-def}
start \emph{polarizing}: they approach either a \emph{noiseless channel} or a 
\emph{useless channel}. Formally, given a 
BMS~channel~$W$, its \emph{capacity} $I(W)$\textcolor{black}{,} and
\emph{Bhattacharyya parameter} $Z(W)$~are~defined~by
\begin{equation}\label{Z-def} 
\begin{split}
I(W) &
\,\ \textcolor{black}{\triangleq}\,\
\frac{1}{2}
\sum_{y\in\sY}\hspace{-2pt}\sum_{x\in\bits\hspace{-9pt}}
\hspace{-3pt}W(y|x)\log_2\frac{W(y|x)}{\frac{1}{2}W(y|0) + \frac{1}{2}W(y|1)}\ ,\\
Z(W) &
\,\ \textcolor{black}{\triangleq}\kern1pt
\sum_{y\in\sY} \!\sqrt{W(y|0)W(y|1)}.
\end{split}
\end{equation}
Given $\delta \,{\in}\, (0,1)$, let us say that
a bit-channel $W_i$ is \emph{$\delta$-bad\/} if $Z(W_i) \ge 1-\delta$~and 
\emph{$\delta$-good\/} if $Z(W_i) \le \delta$. Then the polarization
theorem of \Arikan~\cite[Theorem\,1]{Ari09} can be informally
stated~as~follows.
\begin{theorem*}[Polarization theorem]
\label{thm1}
For every $\delta \,{\in}\, (0,1)$, almost all bit-channels become 
either $\delta$-good or $\delta$-bad as $n \to \infty$. In fact,
as $n \to \infty$, the fraction of $\delta$-good bit-channels
approaches the capacity $I(W)$~of~the~underlying channel $W$,
while the fraction of $\delta$-bad bit-channels approaches 
$1-I(W)$. 
\end{theorem*}

With $\delta =  o(1/n)$,
this theorem naturally leads to the construction of capacity-achieving
\emph{polar codes}.
Specifically,~an $(n,k)$ polar code 
is constructed by selecting a set \smash{$\cA$} of $k$ {\textcolor{black}{$\delta$-}good} 
bit-channels to carry the information bits, while the 
input to all the other bit-channels is frozen to zeros.
In practice, 
the code parameters $k$ and~$\delta$ are usually selected according 
to the target rate of the code and/or the desired probability of error.

\looseness=-1
Henceforth, let us focus on the \emph{binary erasure channel} with erasure
probability $z$, which we denote as~$\BEC(z)$. It is well known that
for $W = \BEC(z)$, we have $Z(W) = z$ and $I(W) = 1-z$. It is furthermore known
(see, for example, \cite[Section 3.4]{Hassani}, \cite{F14}, or \textcolor{black}{\cite[Section 2.2]{Fazeli}}) that if $W = \BEC(z)$,
then for all $i \in [n]$, the $i$-th bit-channel $W_i$ is also
a~binary erasure channel $\BEC\bigl(p_i(z)\bigr)$, whose erasure
probability $p_i(z)$ is a polynomial of degree at most $n$ in $z$.

A proof of the polarization theorem for the BEC follows by studying 
the evolution of these $n$ erasure probabilities $p_i(z)$ as 
$n = \ell^m$ grows. For a fixed kernel $K$, this evolution is
completely determined by the erasure probabilities of the $\ell$
bit-channels obtained after a \emph{single step of polarization}.
These $\ell$ erasure probabilities are a central object of study in this paper.
\begin{definition*}[Polarization behavior]
\label{def1}
Let $W = \BEC(z)$ and let $K$ be a fixed $\ell\times\ell$ binary
polarization kernel. For each $i \in [\ell]$, we let $f_{K,i}(z)$ denote
the erasure probability of the bit-channel $W_i$ given by\/
{\rm \eq{Wi-def}} with $n = \ell$\linebreak
and 
$
W^*\kern-1pt(\yyy|\uuu) 
=
W^\ell\kern-1pt
\bigl(\yyy\hspace{1pt}{\bigm|}\hspace{1pt}\uuu\hspace{1pt} K \bigr)
$. 
We refer to the set of $\ell$ polynomials
$\bigl\{f_{K,1}(z),f_{K,2}(z),\ldots,f_{K,\ell}(z)\bigr\}$
as the 
polarization behavior of the kernel $K$.
\end{definition*}
\noindent
Indeed, we shall see later in this paper that 
$f_{K,i}(z)$ is a polynomial of degree at most $\ell$ in $z$,
for all $i$. 
For~example, in the special case of the $2\times2$ kernel \eqref{eq:orker},
the polarization behavior is given by
$f_{K,1}(z) = 2z-z^2$ and $f_{K,2}(z) = z^2$. 
With this notation, it is advantageous to view
the $n = \ell^m$ erasure probabilities $p_i(z)$ 
as the values taken by a random variable $Z_m$ 
induced by the uniform distribution on the $\ell^m$ bit-channels.
Given that \textcolor{black}{$K$} is non-singular, one can show that $\Km$ is also non-singular. Furthermore, by applying the chain rule of mutual information, since the matrix $\Km$ is nonsigular, it is easy to see that the polar transform in Figure~\ref{polar-def} preserves capacity. We can then study the evolution of this random variable $Z_m$ as $m$ grows.
More formally, the recursive construction of $\Km$ 
makes it possible to introduce the martingale
$\{Z_{m}\}_{m\in \mathbb N}$ defined as follows:
\begin{align}
\label{eq:random_process_definition}
Z_{m+1} \, = \, f_{K,B_m}(Z_m) &,~~~~\text{ for }B_m \sim \textsf{Uniform}[\ell],
\end{align}
with the initial condition $Z_0 = z$.
One can view \eq{eq:random_process_definition} as a stochastic process
on an infinite \textcolor{black}{$\ell$-ary} tree, where in each step we take one of the
$\ell$ available branches with uniform probability. The polarization
theorem then follows from the \textcolor{black}{almost sure convergence given by} the martingale convergence theorem,
which in this case implies that 
\textcolor{black}{
\begin{align}
\mathbb{P}\big(\lim_{m\rightarrow \infty}Z_m(1-Z_m) = 0\big) = 1,
\end{align}
where the probability measure is defined with respect to the random selection of the bit-channel indices.} This shows that the erasure probabilities $p_i(z)$ of 
the $\ell^m$ bit-channels polarize to either $0$ or $1$ as $m\to\infty$. 
Hence, the fraction of bit-channels that polarize to $0$ approaches $I(W)$. 
The speed with which this polarization phenomenon takes
place is the determining factor in the decay rate of 
the gap to capacity as a~function of the block length $n = \ell^m$. 
We elaborate on this in the next subsection.

\vspace{3.0ex}
\subsection{On the rate of polarization in various regimes}
\label{subsec:polar}

\looseness=-1
The performance of polar codes has been analyzed in several
regimes. In the \emph{error-exponent} regime, the rate $R<I(W)$ is
fixed, and we study how the error probability $\Pe$ scales as 
a~function of the block length $n$. This 
is represented by the
vertical/blue cut in Figure\,\ref{fig:scalingreg}. In \cite{ArT09}, 
it is shown that the error probability under successive-cancellation 
decoding behaves roughly as $2^{-\sqrt{n}}$. A more refined scaling 
in this regime is proved in~\cite{HMTU13}.

\begin{figure}[t] 
\centering 
\includegraphics[width=0.8\columnwidth]{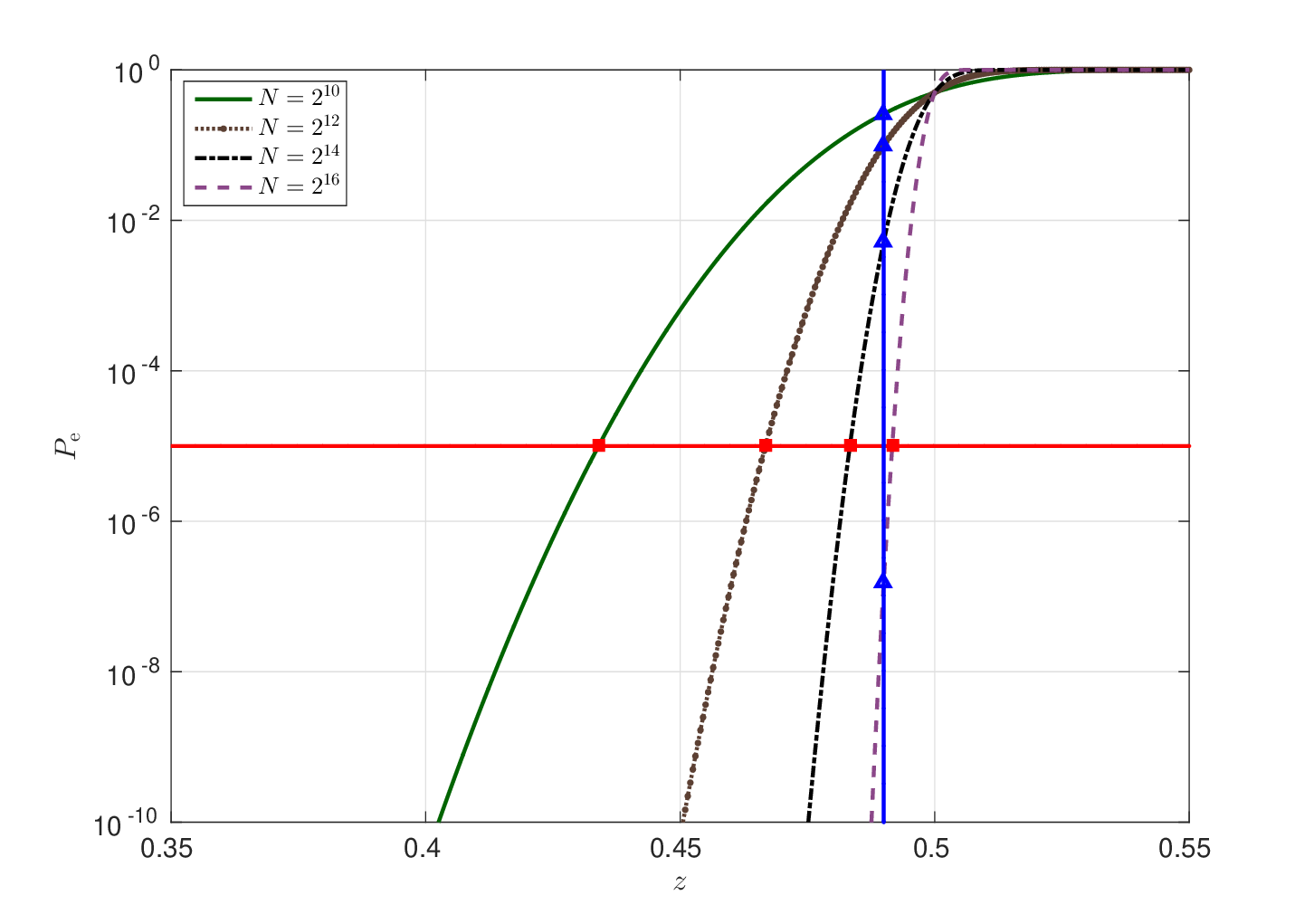}
\caption{Performance of a family of codes with rate $R=0.5$ \textcolor{black}{over binary erasure channels}. Different
curves correspond to different codes of varying block length $n$.
The $x$-axis is the \textcolor{black}{erasure probability of  underlying transmission channel, denoted by $z$}, and the $y$-axis is the error probability $\Pe$. The error-exponent regime captures the behavior of the blue/vertical cut at a fixed channel parameter $z$ (or, equivalently, at a fixed gap to capacity $I(W)-R$). The error-floor regime captures the behavior of a single curve, of fixed block length $n$. The scaling-exponent regime captures the behavior of the red/horizontal cut at a fixed error probability $\Pe$. The figure is courtesy of \cite{MHU2}.}
\label{fig:scalingreg}
\end{figure} 

\textcolor{black}{As common practice for comparison purposes, we consider using a communication channel over a family of channels that can be characterized by a single channel parameter such as the erasure probability in BEC.} In the \emph{error-floor} regime, the code is fixed (i.e., the rate
$R$ and the block length $n$ are fixed), and we study how the
error probability $\Pe$ scales as a function of the channel
parameter. This approach corresponds to taking into account one of the
four curves in Figure\,\ref{fig:scalingreg}. In~\cite{EP13}, it is
proved that the stopping distance of polar codes scales as $\sqrt{n}$,
which implies good error-floor performance under belief-propagation decoding. 
The authors~of~\cite{EP13} also provide simulation results that show 
no sign of an error floor for transmission over the BEC and over the
binary-input AWGN channel. This problem is completely settled in \cite{MHU2},
where it is shown that polar codes do not exhibit error floors under
transmission over any BMS channel.

\looseness=-1
The focus of this paper is on the \emph{scaling-exponent} regime,
where the error probability $\Pe$ is fixed, and we study how the
gap to capacity $I(W)-R$ scales as a function of the block length $n$. 
This approach is represented by the horizontal/red cut in
Figure\,\ref{fig:scalingreg}. As mentioned earlier, if $n$ 
is $O\bigl(1/(I(W)-R)^\mu\bigr)$, we say that the family
of codes has \emph{scaling exponent} $\mu$. For polar codes, the 
value of $\mu$ depends on the underlying channel~$W$. In
\cite{KMTU}, a heuristic method is presented for computing the
scaling exponent in the case of transmission over the BEC under
successive-cancellation decoding; this method yields $\mu \approx 3.627$. 
In \cite{Guru-Xia, Guru-Xia2}, it is shown that the block length,
construction, encoding and decoding complexity are all bounded by a
polynomial in the inverse of the gap to capacity, for transmission
over any BMS channel. This implies that there exists a finite scaling
exponent~$\mu$. Rigorous bounds on $\mu$ are provided in 
\cite{HAU,Goldin-Burshtein,MHU2}. In \cite{HAU}, it is proved that 
$3.579 \le \mu \le 6$, and it is conjectured that the lower bound can be
increased to $3.627$ (i.e., up to the value heuristically computed
for the BEC). In \cite{Goldin-Burshtein}, the upper bound is improved 
to $\mu \le 5.702$. The currently best-known upper bounds on the 
scaling exponent are established
in \cite{MHU2}: for any BMS channel, $\mu \le 4.714$; and for
the special case of the BEC, $\mu\le 3.639$, which approaches the
value obtained heuristically in \cite{KMTU}. As a side note, let us
point out that the heuristic method of \cite{KMTU} is based on 
a ``scaling assumption'' which requires the existence of a certain
limit. The results of \cite{HAU, Goldin-Burshtein, MHU2}, as well as
the results presented in this paper, do not rely on such an assumption.

In a nutshell, the scaling exponent of \textcolor{black}{classical polar codes constructed via Ar{\i}kan's \textcolor{black}{$2 \times 2$ kernel}} 
is~around~$4$. Its exact value depends on the underlying transmission channel 
and can be bounded as $3.579 \le \mu \le 4.714$. In contrast, random
binary linear codes achieve the optimal scaling exponent of $2$. 
\pagebreak[3.99]
This means that, in order to obtain the same gap to capacity, the block 
length of polar codes needs to be roughly the square of the block length 
of random codes. Hence, a natural question is how to improve the scaling
exponent of polar codes.

\looseness=-1
One possible approach is to improve the successive-cancellation
decoding algorithm. In particular,~the~succes\-sive~cancellation 
\emph{list decoder} proposed in \cite{TVlist} empirically provides 
a significant improvement in performance.
However, \cite{MHU} establishes a negative result
for list decoders: the introduction of any finite-size list cannot improve
the scaling exponent under MAP decoding for transmission over any
BMS channel. Furthermore,~for the special case of the BEC, it is also
proved in \cite{MHU} that the scaling exponent under successive-cancellation
decoding does not change even under a finite number of interventions
(that reverse incorrect decisions) from a genie.

\looseness=-1
Another approach is to consider polarization kernels of size 
larger than \Arikan's $2\times 2$ matrix \eqref{eq:orker}. 
Indeed, it is already known that such
kernels have the potential to improve the scaling behavior of polar
codes. For the error-exponent regime, Korada, \c{S}a\c{s}o\u{g}lu, 
and Urbanke proved in~\cite{KSU} that for $\ell$ sufficiently large,
there exist $\ell\times\ell$ binary kernels such that the
error probability of the resulting polar codes scales roughly as $2^{-n}$,
rather than~$2^{-\sqrt{n}}$. For the scaling-exponent regime, Fazeli
and Vardy~\cite{F14} observed that the value of $\mu$ on the BEC
can be reduced from $\mu = 3.627$ for the matrix in \eq{eq:orker}
to $\mu(K_8) = 3.577$
and
$\mu(K_{16}) = 3.356$,
where $K_8$ and $K_{16}$ are specific binary kernels constructed 
in~\cite{F14}.
Pfister and Urbanke~\cite{PU16} recently proved that, in the case of 
transmission over the $q$-ary erasure channel, the optimal 
scaling-exponent value of $\mu = 2$ can be approached as both
the size of the kernel $\ell$ and the size of the alphabet $q$
grow without bound. Furthermore, Hassani~\cite{Hassani}
gives evidence supporting the conjecture that, in order to
approach $\mu=2$ on the erasure channel, it suffices to consider large kernels 
over the \emph{binary alphabet}. 
%
Herein, we finally settle this conjecture.


\vspace{3.0ex}
\section{Outline of the Proof}
\label{sec:outline}

The proof of our main result consists of several major steps. 
The technical part of the proof is, on occasion, quite intricate.
To help the reader, we briefly discuss the main ideas behind each 
of the steps in this section.

\vspace{1.08ex}
\noindent
{\bf Step\,1: Characterization of the polarization process.} 
In order to understand the finite-length scaling of polar codes, 
we need to understand how fast the random process $Z_m$ 
defined in \eq{eq:random_process_definition}
polarizes. In other words, given a~small 
\textcolor{black}{$\zeta > 0$}, how fast does the quantity 
$\mathbb{P}\{Z_m \in [\zeta, 1-\zeta]\}$ 
vanish with $m$? To answer this question, we first relate 
the decay rate of $Z_m$ with another quantity that
can be directly computed from the kernel matrix $K$.

\looseness=-1
As the first step along these lines, we consider the behavior
of another random process $Y_m = g_\al(Z_m)$, where
$g_{\alpha}(z) = z^\al(1-z)^\al$, and $\alpha > 0$ is 
a parameter to be determined later. Note that 
$Z_m \!\in [\textcolor{black}{\zeta}, 1-\textcolor{black}{\zeta}]$ if and only~if~$Y_m$
is lower-bounded by $\textcolor{black}{\zeta}^\al(1-\textcolor{black}{\zeta})^\al$. 
Therefore, by Markov's inequality, we have
\begin{align} 
\label{markoveps}
{\mathbb P} \bigl\{Z_m \in [\textcolor{black}{\zeta},1-\textcolor{black}{\zeta}] \bigr\} 
\ \leq \
\frac{\mathbb{E}[g_{\alpha}(Z_m)]}{\textcolor{black}{\zeta}^\al(1-\textcolor{black}{\zeta})^\al} 
\end{align}
In order to derive an upper bound on $\mathbb{E} [g_{\alpha}(Z_m)]$, 
we write:
\begin{align}
\begin{split}
g_{\alpha}(Z_m) 
& =\ 
\Bigl(f_{K,B_m} (Z_{m-1}) \bigl(1-f_{K,B_m} (Z_{m-1})\bigr)\Bigr)^\alpha 
\\
& = \
Z_{m-1}^\alpha (1-Z_{m-1})^\alpha  
\left( 
\frac{f_{K,B_m}(Z_{\textcolor{black}{m-1}}) \bigl(1-f_{K,B_m} (Z_{\textcolor{black}{m-1}})\bigr)}{Z_{m-1}(1-Z_{m-1})} 
\right)^\alpha  
\\ 
& =\ 
g_{\alpha}(Z_{m-1})  
\left( 
\frac{f_{K,B_m}(Z_{\textcolor{black}{m-1}}) \bigl(1-f_{K,B_m} (Z_{\textcolor{black}{m-1}})\bigr)}{Z_{m-1}(1-Z_{m-1})} 
\right)^\alpha 
.
\end{split}
\end{align}
Proceeding along these lines, we eventually conclude that
\begin{equation} 
\label{zeta1}
\mathbb{E} [g_{\alpha}(Z_m)] \ \leq \ \bigl(\lambda_{\alpha, K}^*\bigr)^m,
  \end{equation}
  where
  \begin{equation} \label{zeta2}
   \lambda_{\alpha, K}^* 
\ \triangleq \   
\sup_{z \in (0,1)} \frac{1}{\ell} ~
\frac{\displaystyle\sum_{i=1}^{\ell} 
\Bigl(f_{K,i}(z)\bigl(1-f_{K,i}(z)\bigr)\Bigr)^\alpha}{\bigl(z(1-z)\bigr)^\alpha}.
 \end{equation}

\textcolor{black}{The discussion above is \textcolor{black}{presented formally in}~\Lref{lemma:Q_n_upperbound}.}

\vspace{2.70ex}
\noindent
{\bf Step\,2: Sharp transitions in the polarization behavior.} 
\textcolor{black}{We fix $\alpha = 1/\log \ell$ and show that as $\ell$ grows, with probability at least $1-\textcolor{black}{o_{\ell}(1)}$ over the random choice of a non-singular $\ell\times\ell$ binary kernel $K$,
we have
\begin{equation}
\label{lambda-bound}
\lambda_{\alpha, K}^* \ = \ O(\ell^{ -\sfrac{1}{2} }\log \ell) \ .
\end{equation}}%
To do so, we prove that, as $\ell$ grows, the 
polarization-behavior polynomials $f_{K,i}(z)$ will ``look like''
step functions for most nonsingular kernels. First note
that $f_{K,i}(z)$ is an increasing polynomial with 
$f_{K,i}(0) = 0$ and $f_{K,i}(1) = 1$, for any $i$ and any $K$. 
As $\ell$ increases, we show that $f_{K,i}(z)$ is likely to have
a sharp transition~threshold around the point $z = i/\ell$. 
More precisely, we prove that
\begin{equation}
\label{sharpness}
\begin{split}
f_{K,i}(z) &\leq \ell^{-(2+\log \ell)}, \quad\quad\quad  
\mbox{for }z \leq \frac{i}{\ell} -  \textcolor{black}{c_5}\ell^{-1/2}\log \ell,
\\
f_{K,i}(z) &\geq 1- \ell^{-(2+\log \ell)}, \quad \mbox{ for }z\geq \frac{i}{\ell} + \textcolor{black}{c_5} \ell^{-1/2}\log \ell,
\end{split}
\end{equation}
with probability at least $1-O(1/\ell)$ over the random choice of $K$, \textcolor{black}{where $c_5$ is a universal constant}. This threshold behavior is illustrated (both schematically and for certain specific kernels of size $\ell = 16$) in Figure~\ref{fig:large_kernels}.

\begin{figure}[t!]
\centering
\includegraphics[width=0.90\textwidth]{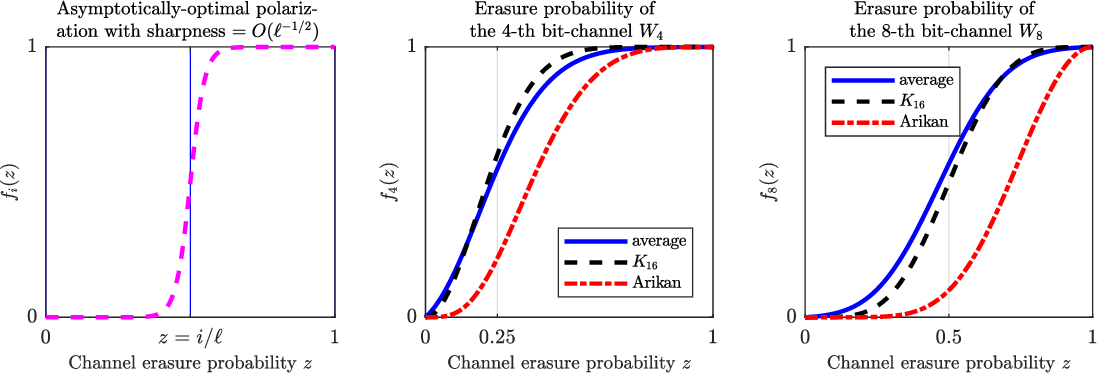}
\caption{The figure on the left illustrates the fact 
that $f_{K,i}(z)$ has a sharp transition of order 
roughly $O\bigl(\ell^{-1/2}\bigr)$ when the kernel $K$ is
chosen at random. The two figures on the right compare three different
choices~of~the~kernel: the red curve corresponds to \Arikan's
kernel; the black curve to the kernel $K_{16}$ from~\cite{F14},
and the blue curve is obtained by taking the average of 
the functions $f_{K,i}(z)$ for all nonsingular $16\times16$ kernels.}
\label{fig:large_kernels}
\end{figure}

\vspace{2.70ex}
\noindent
{\bf Step\,3: Finite-length scaling law.} 
\textcolor{black}{For any fixed $\delta> 0$}, we can derive the finite-length scaling law for polar codes using the results of the previous two steps. 
From \eqref{markoveps}, \eqref{zeta1}, and \eq{lambda-bound},
we conclude that
\begin{equation}
\label{14}
{\mathbb P} \bigl\{Z_m \in [\textcolor{black}{\zeta},1-\textcolor{black}{\zeta}] \bigr\}
\ = \
O\Bigl( {\zeta}^{-\alpha} \bigl(\ell^{\textcolor{black}{-1/(2+\delta)}}\bigr)^m\Bigr).
\end{equation}
Denote the desired \textcolor{black}{block error probability by $\Pe$},
and set $\epsilon = \Pe/n = \Pe \ell^{-m}$ in \eq{14}, \textcolor{black}{as is common in the polarization theory literature due the union upper bound on the block error rate.} Then we have
\begin{equation}
\label{eq:outlinefin}
{\mathbb P} \bigl\{Z_m \in [\Pe \ell^{-m} ,1-\Pe \ell^{-m} ] \bigr\}
\ =  \ 
O\bigl( \ell^{-m/(2+\delta)}\bigr)
\end{equation}
The foregoing is an upper bound on 
the fraction of bit-channels that are not yet sufficiently
polarized after $m$ polarization steps. Later, we will
also provide a simple bound on the fraction 
${\mathbb P}\{Z_m \ge 1-\Pe \ell^{-m}\}$
of bit-channels that are polarized to the useless state.
Note that if we transmit information only on those bit-channels
whose erasure probability is at most $\Pe/n$, then
a~straightforward union-bound argument shows that the overall
probability of error under successive-cancel\-lation decoding
is at most $\Pe$. In essence, the bound in \eq{eq:outlinefin}
implies that the fraction of such ``good'' bit-channels is at 
least $I(W) - O\bigl(\ell^{-m/(2+\delta)}\bigr)$.
Since the block length $n$ is $\ell^m$, this means
that the gap to capacity 
scales roughly as
\smash{$n^{-1/(2+\delta)}$}, which is
the desired scaling law. \textcolor{black}{\Lref{lemma:step1} captures the above discussed argument.}

\vspace{2.00ex}
\section{Main Result}
\label{sec:main_theorem}
\vspace{-0.25ex}

We begin by specializing \Tref{thm1.2} to polar codes 
and stating this result more precisely in \Tref{thm:main_theorem1}. 
We~then gradually reduce the proof of \Tref{thm:main_theorem1} to more
and more specialized statements about large binary kernels. \textcolor{black}{To do so, we start from the following definitions.}\vspace{-1.80ex}

\textcolor{black}{
\begin{definition}[Polar codes with large kernels]
	Consider transmission~over a binary erasure channel $W$ with capacity $I(W)$. Let $\GL(\ell,\Ftwo)$ denote the general linear group of all $\ell\times\ell$ non-singular matrices over $\Ftwo$. Let $K$ be a specific polarization kernel chosen from $\GL(\ell,\Ftwo)$. We define $\mathcal C_K(n, R, \Pe)$ to be a code of rate $R<I(W)$ obtained by polarizing $K$ whose block length $n=\ell^m$ is the smallest such that the error probability under successive cancellation decoding is at most $\Pe$.
\end{definition}
}

\textcolor{black}{
We observe that the code $\mathcal C_K(n, R, \Pe)$ defined above always exists by the results of \cite{KSU} and \cite[Theorem\,5.4]{Sasoglu-FnT}, as long as the matrix $K$ is not equivalent under column permutations to an upper triangular matrix.
}\vspace{-5.40ex}

\textcolor{black}{
\begin{definition}[Upper bound on the scaling exponent]
	Consider transmission~over a binary erasure channel $W$ with capacity $I(W)$. For a fixed polarization kernel $K\! \in \GL(\ell,\Ftwo)$, we define $\overline{\mu}(K)$ to be an upper bound on the scaling exponent of $K$, if for any probability of error $\Pe \in (0,1)$ and any rate  $R<I(W)$, there exists a~$\mathcal C_K(n, R, \Pe)$ polar code whose block length $n$ is upper bounded by 
	\begin{equation}\label{eq:scalingexp}
		\begin{split}
		n\: \leq \frac{\beta}{(I(W)-R)^{\overline{\mu}(K)}}\hspace{.7ex},
		\end{split}
	\end{equation}
	where $\beta$ is a constant that depends only on $K$ and $\Pe$. 
\end{definition}
Note that these definitions are consistent with the way the scaling exponent is defined in the literature, see e.g., \cite[Theorem 1]{Goldin-Burshtein}, \cite[Theorem 1]{MHU2}.}
\textcolor{black}{We are now ready to present our main results.}\vspace{-2.70ex}

\textcolor{black}{
\begin{theorem}%
[Binary polar codes with near-optimal scaling] 
\label{thm:main_theorem1}
\looseness=-1
Consider transmission~over a binary erasure \linebreak channel $W$ with capacity
$I(W)$. Let $K\! \in \GL(\ell,\Ftwo)$ be a kernel selected uniformly at random. Fix $\delta \in (0, 1]$. Then, there exists $\ell_0(\delta)$ such that for any $\ell>\ell_0(\delta)$, with high probability over the choice of $K$,  the scaling exponent of $K$ is upper-bounded by
\begin{equation}\label{eq:main_theorem_scaling_bound}
\overline{\mu}(K) \leq 2+\delta.
\end{equation}  
Furthermore, the constant $\beta$ in~(\ref{eq:scalingexp}) is given by $\ell(1+2\hspace{0.1em}\Pe^{\textcolor{black}{-1}})^{3}$. 
\end{theorem}
}

In fact, what we prove is slightly stronger. Let $\{\mathcal C_K(n=\ell^m, R)\}_{m=1}^{\infty}$ be the family of rate-$R$ large-kernel polar codes obtained by polarizing $K$ for $m$ many steps, where $R<I(W)$. 
We will show that as $\ell$ grows, with high probability over the choice of $K$, the scaling exponent 
can be upper-bounded by 
\begin{align}\label{eq:real_scaling_exponent}
\mu = 2 + O\left(\frac{\log(\log \ell)}{\log \ell}\right). 
\end{align}

In order to treat~(\ref{eq:real_scaling_exponent}) as $\mu\leq 2+\delta$, it suffices to pick \textcolor{black}{a fixed} $\ell$ that is at least in the order of $\exp{(\delta^{-1.01})}$. We denote this minimum value of $\ell$ by $\ell_0(\delta)$. Once again, we emphasize that $\ell_0(\delta)$ scales exponentially with $\sfrac{1}{\delta}$. Note that a fixed value of $\ell$, although being extremely large, does not change the asymptotic code-length, $n$, nor the decoding complexity, $O(n\log n)$, which is the main focus of this work. However, given how large $\ell$ should be and the fact that the decoding complexity of polar codes with arbitrary $\ell\times\ell$ kernels is also multiplied by $2^{\ell}$, it becomes clear that the large-kernel polar codes, whose kernels are chosen at random, are not suitable for the practical purposes. We address the recent advancements on the decoding problem of large kernels at the end of the paper. 

\textcolor{black}{We also point out that, as the rate $R$ approaches the channel capacity $I(W)$, which consequently makes the block length $n$ grow, these codes have construction complexity $O(n)$ and encoding/decoding complexity\, $O(n\log n)$. The claim on the construction complexity follows from the fact that the erasure probabilities of the bit-channels can be computed exactly according to the recursion \eqref{eq:random_process_definition}. The claim on the encoding/decoding complexity follows from \cite[Section\,VII]{KSU}. 
}

The foregoing theorem follows from the following result that characterizes the behavior of the polarization process defined in \eq{eq:random_process_definition}.

\begin{theorem}[\textcolor{black}{Near-optimal} scaling of the polarization process]
\label{thm:polarizationscaling}
Let $K \!\in \GL(\ell,\Ftwo)$ be a kernel selected uniformly~at random \textcolor{black}{from} all $\ell\times\ell$ nonsingular binary matrices. Let $Z_m$ be the random process defined in \eqref{eq:random_process_definition} with initial condition $Z_0 = z$. Fix $\Pe \in (0, 1)$ and a small constant $\delta > 0$. Then, there exists $\ell_0(\delta)$ \textcolor{black}{such that for all $\ell>\ell_0(\delta)$ and for all $m\ge 1$, and almost all $K$}, we have
\begin{align}\label{eq:connection_lemma}
\mathbb{P}\{Z_m \le \Pe\ell^{-m}\} 
\ \geq \
1-z \,-\, \bigl(1+2\hspace{0.1em}\Pe^{\textcolor{black}{-1}}\bigr)\ell^{-\frac{m}{\textcolor{black}{2+\delta}}}.
\end{align}
\end{theorem}

\textcolor{black}{For the sake of clarity, note that  \emph{(i)} in \eqref{eq:connection_lemma} the kernel $K$ is fixed and the probability space is defined with respect to the random process $Z_m$, and \emph{(ii)} the result \eqref{eq:connection_lemma} holds with high probability over the
choice of the kernel $K$. We are now ready to present the proof of
\Tref{thm:main_theorem1}.}

\begin{proof}[Proof of \Tref{thm:main_theorem1}]
\hspace*{-0.54ex} \textcolor{black}{Fix any rate $R$ with $R<I(W)$.} Assuming \Tref{thm:polarizationscaling} 
holds, consider transmission
over $\BEC(z)$ of a polar code with block length $n=\ell^m$ and rate $R$ 
obtained by polarizing the $\ell\times \ell$ kernel $K$, where
$\ell>\ell_0(\delta)$. 
By \Tref{thm:polarizationscaling}, \textcolor{black}{with high probability over the choice of $K$,}
~at~least 
a 
\smash{$1-z-(1+2\hspace{0.1em}\Pe^{\textcolor{black}{-1}})\ell^{-{m}/{\textcolor{black}{2+\delta}}}$} fraction 
of the bit-channels have erasure probability at most $\Pe \ell^{-m}$.
\textcolor{black}{Given that $R < 1-z$, one can always find a positive integer $m$ such that}
\begin{equation}\label{eq:ratefin}
\textcolor{black}{R \: \leq \: 1-z-(1+2\hspace{0.1em}\Pe^{\textcolor{black}{-1}})\ell^{-\frac{m}{\textcolor{black}{2+\delta}}}}.
\end{equation}
A simple union bound yields that the error probability under successive cancellation decoding is at most $\Pe$. \textcolor{black}{Given that $I(W)=1-z$, we can re-arrange \eqref{eq:ratefin} to obtain that}
\begin{align}
\textcolor{black}{
(1+2\Pe^{-1})\ell^{-\frac{m}{\textcolor{black}{2+\delta}}}\leq I(W)-R,
}
\end{align}
\textcolor{black}{which is equivalent to}
\begin{align}
\textcolor{black}{
\frac{1+2\Pe^{\textcolor{black}{-1}}}{\big(I(W)-R\big)} \:\leq\: n^{\frac{1}{\textcolor{black}{2+\delta}}} \: \Leftrightarrow \frac{(1+2\Pe^{\textcolor{black}{-1}})^{\textcolor{black}{2+\delta}}}{\big(I(W)-R\big)^{\textcolor{black}{2+\delta}}} \leq n.
}
\end{align}
W.l.o.g., we can assume that $\delta<1$ and, hence, we can take $\beta$
as prescribed in \Tref{thm:main_theorem1}. \textcolor{black}{Considering that $n$ is a power of $\ell$, for \eqref{eq:ratefin} to hold, it suffices to have}
\begin{align}\label{eq:newb}
\textcolor{black}{
\ceil{\log_{\ell}\bigg(\frac{\beta}{\big(I(W)-R\big)^{\textcolor{black}{2+\delta}}}\bigg)}
\leq \log_{\ell} n\;,  
}
\end{align}
\textcolor{black}{where $\ell$ can be an arbitrary integer such that $\ell \geq \ell_0$. Therefore, the smallest value of $n$ for which the desired code exists is the first integer power of \textcolor{black}{$\ell$} that is not smaller than $\sfrac{\beta}{(I(W)-R)^{2+\delta}}$. Thus, there exists a code with
\begin{align}
n \leq \frac{\beta\textcolor{black}{\ell}}{(I(W)-R)^{2+\delta}}
\end{align}
such that \eqref{eq:newb} holds.
}
\end{proof}

The rest of the section is devoted to the proof of 
\Tref{thm:polarizationscaling}. The basic idea is to bound the
number~of~unpolarized bit-channels. To this end, let us introduce the
polarization measure function $g_{\alpha}(z)$, defined as follows:
\begin{align}
\label{eq:g_alpha_definition}
g_{\alpha}(z) 
\ \triangleq \ 
z^\al(1-z)^{\alpha},
\end{align}
where $\alpha\in(0, 1)$ is a constant parameter to be determined later.
The first step is to show that an upper bound on 
$\mathbb E[g_{\alpha}(Z_m)]$ 
yields a lower bound on 
$\mathbb{P}\{Z_m \le \Pe\ell^{-m} \} $. 
This is accomplished in the following lemma.

\begin{lemma}
\label{lemma:step1}
Let $K\!\in \GL(\ell,\Ftwo)$ be an $\ell\times\ell$ nonsingular binary kernel such that none of its column permuta\-tions is upper triangular. Let $Z_m$ be the random process defined in \eqref{eq:random_process_definition} with initial condition $Z_0 = z$. Fix~a~constant $\alpha\in (0, 1)$ and define $g_{\alpha}(z)$ as in \eqref{eq:g_alpha_definition}. 
Further fix $\textcolor{black}{\rho > 0},\Pe \in (0, 1)$ and assume that
\begin{equation}
\mathbb E[g_{\alpha}(Z_m)]\le \ell^{-m\rho}.
\end{equation}
for all $m\ge 1$. Then, for any $m\ge 1$, we have
\begin{align}
\mathbb{P}\{Z_m \le \Pe\ell^{-m} \} 
\ \geq \ 1-z \ - \ \bigr(2 \Pe^{-\alpha}+\Pe\bigl)\ell^{-m(\rho-\alpha)}.
\end{align}
\end{lemma}

\begin{proof}
First of all, 
we upper bound ${\mathbb P}\{Z_m \in \left[\Pe\hspace{0.1em} \ell^{-m}, 1-\Pe\hspace{0.1em} \ell^{-m}\right]\}$ as follows: 
\begin{equation}
\label{eq:B}
\begin{split}
{\mathbb P}\!\left\{Z_m \in \left[\Pe\hspace{0.1em} \ell^{-m}, 1-\Pe\hspace{0.1em} \ell^{-m}\right]\right\}
&\,\stackrel{\mathclap{\mbox{\footnotesize(a)}}}{=} 
\:{\mathbb P}\!\left\{g_{\alpha}(Z_m) \ge g_{\alpha}(\Pe\hspace{0.1em}\ell^{-m})\right\}
\\
&\stackrel{\mathclap{\mbox{\footnotesize(b)}}}{\le} 
\:\frac{{\mathbb E}[g_{\alpha}(Z_m)]}{g_{\alpha}(\Pe\hspace{0.1em} \ell^{-m})}
\\
&\stackrel{\mathclap{\mbox{\footnotesize(c)}}}{\le} 
\: \frac{\ell^{-m\rho}}{g_{\alpha}(\Pe\hspace{0.1em} \ell^{-m})}
\\
&\stackrel{\mathclap{\mbox{\footnotesize(d)}}}{\le} 
\:2\hspace{0.1em} \Pe^{-\alpha} \hspace{0.1em} \ell^{-m(\rho-\alpha)},
\end{split}
\end{equation}
where equality (a) uses the concavity of 
$g_\alpha(\cdot)$ \textcolor{black}{together with its symmetry around $1/2$}; inequality (b) follows from Markov's inequality;
inequality (c) uses the hypothesis ${\mathbb E}[g_\alpha(Z_m)]\le
\ell^{-m\rho}$; and inequality (d) uses the fact that
$1-\Pe\hspace{0.1em}\ell^{-m}\ge 1/2$ for all $m\ge 1$.
Now, let us \textcolor{black}{fix $m$ from this point forward and} define
\begin{equation}\label{eq:defABC}
\begin{split}
A\: &\triangleq\ {\mathbb P}\left\{Z_m \in \left(0, \Pe\hspace{0.1em} \ell^{-m}\right)\right\},\\ 
B\: &\triangleq\ {\mathbb P}\left\{Z_m \in \left[\Pe\hspace{0.1em} \ell^{-m}, 1-\Pe\hspace{0.1em} \ell^{-m}\right]\right\},\\ 
C\: &\triangleq\ {\mathbb P}\left\{Z_m \in \left(1-\Pe\hspace{0.1em} \ell^{-m}, 1\right)\right\},\\ 
\end{split}
\end{equation}
and let $A'$, $B'$, and $C'$ be the fraction of bit-channels in $A$, $B$, and $C$, respectively, that will have a vanishing erasure probability as $n\to \infty$. More formally, we define
\begin{equation}\label{eq:Aprime}
\begin{split}
A' &\triangleq\ \lim_{m'\to \infty} {\mathbb P}\left\{Z_m \in \left(0, \Pe\hspace{0.1em} \ell^{-m}\right), Z_{m+m'}\le \ell^{-m'}\right\},\\ 
B' &\triangleq\ \lim_{m'\to \infty} {\mathbb P}\left\{Z_m \in \left[\Pe\hspace{0.1em} \ell^{-m}, 1-\Pe\hspace{0.1em} \ell^{-m}\right], Z_{m+m'}\le \ell^{-m'}\right\},\\ 
C' &\triangleq\ \lim_{m'\to \infty} {\mathbb P}\left\{Z_m \in \left(1-\Pe\hspace{0.1em} \ell^{-m}, 1\right), Z_{m+m'}\le \ell^{-m'}\right\}.\\ 
\end{split}
\end{equation} 
\textcolor{black}{We now show that the limits in~(\ref{eq:Aprime}) exists, and consequently the quantities $A'$, $B'$ and $C'$ are well defined. To do so, it suffices to prove that, for any $z\in (0, 1)$, the following limit exists
\begin{equation}\label{eq:limnew1}
\lim_{m'\to \infty} {\mathbb P}\left\{Z_{m+m'}\le \ell^{-m'}\mid Z_m=z\right\},
\end{equation}
which is equivalent to proving the existence of 
\begin{equation}\label{eq:limnew2}
\lim_{m'\to \infty} {\mathbb P}\left\{Z_{m'}\le \ell^{-m'+m}\mid Z_0=z\right\}.
\end{equation}
Note that, for any $\beta>0$,
\begin{equation}
\liminf_{m'\to \infty} {\mathbb P}\left\{Z_{m'}\le \ell^{-m'+m}\mid Z_0=z\right\}\ge \liminf_{m'\to \infty} {\mathbb P}\left\{Z_{m'}\le 2^{-\ell^{m'\beta}}\mid Z_0=z\right\}.
\end{equation}
From the proof of \cite[Theorem 2]{KSU}, we have that the random variable $Z_{m'}$ converges almost surely to a $\{0,1\}$ random variable $Z_{\infty}$. Furthermore, an application of \cite[Lemma 5.9]{Sasoglu-FnT} gives that, for $\beta<E(K)$,
\begin{equation}
\liminf_{m'\to \infty} {\mathbb P}\left\{Z_{m'}\le 2^{-\ell^{m'\beta}}\mid Z_0=z\right\}={\mathbb P}\left\{Z_\infty=0\right\},
\end{equation}
where $E(K)>0$ for all $\ell \times \ell$ nonsingular binary matrices none of whose column permutations is upper triangular. Suppose now that
\begin{equation}
\limsup_{m'\to \infty} {\mathbb P}\left\{Z_{m'}\le \ell^{-m'+m}\mid Z_0=z\right\}>{\mathbb P}\left\{Z_\infty=0\right\}.
\end{equation}
Then, $Z_{m'}$ cannot converge in probability to $Z_\infty$ \cite[Page 70, Eq. (5)]{chung}. However, this cannot be possible since almost sure convergence implies convergence in probability \cite[Theorem 4.1.2]{chung}. Thus, 
\begin{equation}
\limsup_{m'\to \infty} {\mathbb P}\left\{Z_{m'}\le \ell^{-m'+m}\mid Z_0=z\right\}=\liminf_{m'\to \infty} {\mathbb P}\left\{Z_{m'}\le \ell^{-m'+m}\mid Z_0=z\right\},
\end{equation} 
and the limits in \eqref{eq:limnew1} and \eqref{eq:limnew2} exist.
}
Note that 
\begin{equation}\label{eq:sumABC}
A' + B' +C' 
~=~
\lim_{m'\to \infty} {\mathbb P}\left\{Z_{m+m'}\le \ell^{-m'}\right\}
~=~
1-z.
\end{equation}
In addition, from \eqref{eq:B} we have that
\begin{equation}\label{eq:Bprime}
B' \le B \le 2\hspace{0.1em} \Pe^{-\alpha} \hspace{0.1em} \ell^{-m(\rho-\alpha)}.
\end{equation}
In order to upper bound $C'$, we proceed as follows:
\begin{equation}\label{eq:Cprimecalc}
\begin{split}
C' &= \lim_{m'\to \infty} {\mathbb P}\left\{Z_{m+m'}\le \ell^{-m'}\mid Z_m \in \left(1-\Pe\hspace{0.1em} \ell^{-m}, 1\right]\right\} \cdot {\mathbb P}\left\{Z_m \in \left(1-\Pe\hspace{0.1em} \ell^{-m}, 1\right]\right\} \\
&\le \lim_{m'\to \infty} {\mathbb P}\left\{Z_{m+m'}\le \ell^{-m'}\mid Z_m \in \left(1-\Pe\hspace{0.1em} \ell^{-m}, 1\right]\right\}.
\end{split}
\end{equation}
By using again the fact that the kernel $K$ is polarizing, we obtain that the last term equals the capacity of a BEC with erasure probability at least $1-\Pe\hspace{0.1em} \ell^{-m}$. Consequently,
\begin{equation}\label{eq:Cprime}
C' \le \Pe\hspace{0.1em} \ell^{-m}.
\end{equation}  
As a result, we conclude that
${\mathbb P}\{Z_m \in \left[0, \Pe\hspace{0.1em} \ell^{-m}\right)\} = A$ 
is bounded as follows
\begin{align*}
A \,\ge\, A' 
~\stackrel{\mathclap{\mbox{\footnotesize(a)}}}{=}~
1-z- B' - C' 
~\stackrel{\mathclap{\mbox{\footnotesize(b)}}}{\ge}~ 
1-z - 2\hspace{0.1em} \Pe^{-\alpha} \hspace{0.1em} \ell^{-m(\rho-\alpha)} - \Pe\hspace{0.1em} \ell^{-m}
~\stackrel{\mathclap{\mbox{\footnotesize(c)}}}{\ge}~
1-z - \left(2\hspace{0.1em}\Pe^{-\alpha}+\Pe\right)\ell^{-m(\rho-\alpha)},
\end{align*}
where equality (a) uses \eqref{eq:sumABC}; inequality (b) uses \eqref{eq:Bprime} and \eqref{eq:Cprime}; and inequality (c) uses the fact that \textcolor{black}{since $\alpha,\rho \in(0,1)$ then $\rho-\alpha <1$}. This chain of inequalities implies the desired result.
\end{proof}

The second step is to derive an upper bound on $\mathbb
E[g_{\alpha}(Z_m)]$ of the form $(\lambda_{\alpha,K}^*)^{m}$, where
$\lambda_{\alpha,K}^*$ depends on the particular kernel $K$. This is
accomplished in \Lref{lemma:Q_n_upperbound}, whose statement and proof
follow.

\begin{lemma}
\label{lemma:Q_n_upperbound}
Let $K\!\in \GL(\ell,\Ftwo)$ be a fixed $\ell\times\ell$ binary kernel.
Let $Z_m$ be the random process defined 
in \eqref{eq:random_process_definition}~with initial condition $Z_0 = z$. Fix $\alpha\in (0, 1)$ and define $g_{\alpha}(z)$ as in \eqref{eq:g_alpha_definition}. For $z\in (0, 1)$, define $\lambda_{\alpha,K}(z)$  as
\begin{align}\label{eq:lambda_definition}
\lambda_{\alpha,K}(z) \triangleq \frac{\frac{1}{\ell}\sum_{i=1}^{\ell}g_{\alpha}(f_{K,i}(z))}{g_{\alpha}(z)},
\end{align}
and let $\lambda_{\alpha,K}^*$ be its supremum, i.e., 
\begin{align}
\label{eq:lambda_definition2}
\lambda_{\alpha,K}^*  \triangleq\underset{z\in(0,1)}\sup \lambda_{\alpha,K}(z).
\end{align}
Then, for any $m\ge 0$, we have that
\begin{align}\label{eq:Q_n_upperbound_lambda}
\mathbb E[g_{\alpha}(Z_m)]\leq (\lambda_{\alpha,K}^*)^{m}g_{\alpha}(z). 
\end{align}
\end{lemma}

\begin{proof}
We prove the claim by induction. The base step $m=0$ follows
immediately from the fact that $Z_0=z$. To prove the inductive step,
we write
\begin{align}
\mathbb E\big[g_{\alpha}(Z_{m+1})\big] &= \mathbb{E}\big[\mathbb{E}[g_{\alpha}(f_{\textcolor{black}{K,B_m}}(Z_{m})) \mid Z_m]\big],
\end{align}
where the first (outer) expectation on the RHS is with respect to $Z_m$ 
and the second (inner) expectation is with respect to $B_m$. Then, we have 
that 
\begin{align}
\begin{split}
\mathbb{E}\big[\mathbb{E}[g_{\alpha}(f_{\textcolor{black}{K,B_m}}(Z_{m})) \mid Z_m]\big] 
~&=~
\mathbb{E}\bigg[ g_{\alpha}(Z_{m}) \frac{\frac{1}{\ell}\sum_{i=1}^{\ell}g_{\alpha}(f_{K,i}(Z_m))}{g_{\alpha}(Z_{m})}\bigg]
\\
~&\leq~ \mathbb E\big[g_{\alpha}(Z_{m})\big] \!\!\underbrace{\underset{z\in\{0,1\}}{\text{sup}}\frac{\frac{1}{\ell}\sum_{i=1}^{\ell}g_{\alpha}(f_{K,i}(z))}{g_{\alpha}(z)}}_{\lambda_{\alpha,K}^*}.
\end{split}
\end{align}
\end{proof}

The third and final step is to prove that $\lambda_{\alpha,K}^*$
concentrates around $1/\sqrt{\ell}$, when $K$ is selected uniformly at
random among all $\ell\times\ell$ nonsingular binary matrices. This
is done in \Tref{thm:main_theorem2}, which is stated below.

\textcolor{black}{
\begin{theorem}[Concentration of $\lambda_{\alpha,K}^*$]
\label{thm:main_theorem2}
Let $K\!\in \GL(\ell,\Ftwo)$ be a kernel selected uniformly at random
among all $\ell\times\ell$ non-singular binary matrices. Set
$\alpha={1}/{\log\ell}$ and define $\lambda_{\alpha,K}^*$ as in
\eqref{eq:lambda_definition2}. Then, there exists a universal constant $c$ such that as $\ell$ grows, we have 
\begin{align}\label{eq:ineq_lambda_main}
\mathbb{P} 
\left\{ \lambda_{\alpha,K}^* \leq c\ell^{-\frac{1}{2}}\log \ell\right\}
~\geq~
1-o(1)\hspace{1.5ex},
\end{align}
where the probability space is defined over the choice of the kernel $K$.  
\end{theorem}}

At this point, we are ready to put everything together and present
a proof of \Tref{thm:polarizationscaling}, assuming that 
\Tref{thm:main_theorem2} holds. The proof of \Tref{thm:main_theorem2}
is deferred to the next section.

\begin{proof}[Proof of Theorem \ref{thm:polarizationscaling}]
A simple counting over the binary subspaces of dimension $\ell$ shows that 
\begin{align}
\frac{\text{number of \textcolor{black}{non-singular} but non-polarizing $\ell\times\ell$ binary kernels}}{\text{number of non-singular $\ell\times\ell$ binary matrices}} = \frac{\ell!2^{\frac{\ell(\ell-1)}{2}}}{\prod_{i=0}^{\ell-1}(2^{\ell}-2^{i})} \leq \frac{1}{2^{\ell}} \hspace{1ex}\text{ for all }\hspace{1ex} \ell\geq 5.
\end{align}
\textcolor{black}{Therefore, as $\ell$ grows, with probability at least $1-2^{-\ell} = 1 - o(1)$ over the choice of the kernel $K$,} $K$ is such that none of its column permutations is upper triangular. 
\textcolor{black}{By Theorem \ref{thm:main_theorem2}, as $\ell$ grows, with probability at least $1-o(1)$ over the choice of the kernel $K$, we also have that
\smash{$\lambda_{\alpha,K}^* \leq c\ell^{-1/2}\log\ell$}. Given that the intersection of these two sets also has  probability at least $1-o(1)$,} most choices of $K$ satisfy both conditions for sufficiently large $\ell$. Fix any such kernel. Consequently, as $g_{\alpha}(z)\le 1$ for all $z\in (0, 1)$, by Lemma \ref{lemma:Q_n_upperbound} we have that 
\begin{equation}
\mathbb E[g_{\alpha}(Z_m)]\leq 
\textcolor{black}{
\big(c\ell^{-1/2}\log\ell\big)^m,
}
\end{equation}
where the expectation is over the uniform selection of the polar bit-channel index, or in other words, the random process $Z_m$. \textcolor{black}
{Taking into the account that $\alpha = 1/\log\ell$, we can apply Lemma \ref{lemma:step1} to deduce that
}
\begin{equation}
\label{eq:finstepproc}
\mathbb{P}\{Z_m \le \Pe\ell^{-m} \} ~\geq~1-z-
\textcolor{black}{
c_1(2c\ell^{-1/2}\log \ell)^m,
}
\end{equation}
where $c_1 = 2 \Pe^{-\alpha}+\Pe$ and the probability space is defined with respect to the random selection of the index. Note that, as \textcolor{black}{$\alpha\le 1$} and $\Pe\le 1$, we have that $c_1\le 1+2\Pe^{\textcolor{black}{-1}}$. 
\textcolor{black}{
The theorem immediately follows by picking $\ell_0(\delta)$ to be large enough such that 
\begin{align}
2c\ell_{0}(\delta)^{-1/2}\log \ell_{0}(\delta) \leq \ell_{0}(\delta)^{-\frac{1}{2+\delta}},
\end{align}
which is of the order $O(\exp(\delta^{-1.01}))$.
}
\end{proof}



\vspace{2.70ex}
\section{Proof of Theorem \ref{thm:main_theorem2}: Concentration of $\lambda_{\alpha,K}^*$}\label{appendix:proof}

Recall that our goal is to show that for most non-singular binary kernels 
$K\in\GL(\ell,\Ftwo)$,
\begin{align}\label{eq:maintheorem_target_function}
\lambda_{\alpha,K}(z) \leq \textcolor{black}{
c\ell^{-\frac{1}{2}}\log \ell
}
\hspace{2ex}\forall z\in(0,1),
\end{align}
\textcolor{black}{
where $\alpha=1/\log \ell$ is fixed and $c$ is some universal constant.}
Our strategy is to split the interval $(0,1)$ into the three sub-intervals $(0,1/\ell^2)$, $[1/\ell^2,1-1/\ell^2]$, and $(1-1/\ell^2,1)$. Then, we will show that \eqref{eq:maintheorem_target_function} holds for each of these sub-intervals. In fact, as we shall see, polarization is much faster at the tail intervals. \Pref{thm:cases} captures this approach.

\begin{proposition}
\label{thm:cases}
Let $K\in \GL(\ell,\Ftwo)$ be a kernel selected uniformly \textcolor{black}{at random from} all $\ell\times\ell$ nonsingular binary matrices. \textcolor{black}{Set $\alpha = 1/\log \ell$ and
define  $\lambda_{\alpha,K}(z)$ as in \eqref{eq:lambda_definition}.
}
Then, 
\textcolor{black}{
as $\ell$ grows, 
}
the following results hold.
\begin{enumerate}
\item[\bf 1.]
Near optimal polarization in the middle: 
\textcolor{black}{
\begin{align}\label{eq:cases_middle}
\mathbb{P} \bigg\{\lambda_{\alpha,K}(z)< c\ell^{-\frac{1}{2}}\log\ell, \hspace{0.5em}\forall z\in\Big[\frac{1}{\ell^2},1-\frac{1}{\ell^2}\Big] \bigg\} > 1-o(1), 
\end{align}}
\item[\bf 2.]
Faster polarization at the tails: 
\textcolor{black}{
\begin{align}
\label{eq:cases_tails}
\mathbb{P} \bigg\{\lambda_{\alpha,K}(z)< c\ell^{-1}\log\ell, \hspace{0.5em} \forall z\in\Big(0,\frac{1}{\ell^2}\Big)\cup\Big(1-\frac{1}{\ell^2},1\Big) \bigg\} > 1-o(1), 
\end{align}}
\end{enumerate}
where the probability spaces 
are defined over the choice of the kernel $K$
\textcolor{black}{
 and $c$ is a universal constant.}\pagebreak[3.99]  
\end{proposition}

\begin{proof}[Proof of Theorem~\ref{thm:main_theorem2}]
\textcolor{black}{
Let $A$ and $B$ be the sets of kernels such that the events in~(\ref{eq:cases_middle}) and~(\ref{eq:cases_tails}) respectively hold. Then, by \Pref{thm:cases}, we have that
\begin{equation}\label{eq:aa1}
\begin{split}
 \mathbb{P}(A) &> 1-o(1),\\
 \mathbb{P}(B) &> 1-o(1).
\end{split}
\end{equation}
Furthermore, we have that
\begin{align}\label{eq:aa2}
\mathbb{P}(A\cap B) = \mathbb{P}(A) + \mathbb{P}(B) - \mathbb{P}(A\cup B) \geq \mathbb{P}(A) + \mathbb{P}(B) -1.
\end{align}
By combining \eqref{eq:aa1} and \eqref{eq:aa2}, we conclude that 
\begin{align}\label{eq:result_of_thm_cases}
\mathbb{P} \bigg\{\lambda_{\alpha,K}(z)< c\ell^{-\frac{1}{2}}\log \ell, \hspace{0.5em}\forall z\in(0,1) \bigg\} ~>~ 1-o(1). 
\end{align}}
\vspace*{-1.80ex}
\end{proof}

In what follows, we first analyze the probability of error
under successive-cancellation decoding for the special case of 
transmission over the BEC. \textcolor{black}{We formulate the erasure probability of the $i$-th polar bit-channel (at the kernel level) as a polynomial in the erasure probability of the underlying channel. Then, we utilize this formulation to introduce and compute the average polarization behavior and provide several auxiliary lemmas/propositions to establish the \emph{sharp} transitions of $f_{K,i}(z)$ \emph{on average} that was depicted earlier in Figure~\ref{fig:large_kernels}. Eventually, we put these propositions together and prove a concentration theorem, which, in turn, completes the proof for} \Pref{thm:cases}.

\vspace*{1.80ex}
\subsection{Successive cancellation decoding on binary erasure channels}
\label{subsec:proof_BEC}

\looseness=-1
Let $K\!\in\GL(\ell,\Ftwo)$ be a nonsingular binary kernel,
and let $s$ denote the number of erasures that occurred during 
the transmission over $\BEC(z)$. 
There are a total of $\binom{\ell}{s}$ distinct and equally-likely
erasure patterns, and each of them occurs with probability
$z^s(1-z)^{\ell-s}$. \textcolor{black}{Let ${\eta}_{s}^{(i)}$ denote the number of erasure patterns with $s$ erasures, which make $u_i$ undecodable.} Thus the erasure probability of the $i$-th bit-channel is given by
\begin{align}\label{def:f_K,i(z)}
f_{K,i}(z) = \sum_{s=0}^{\ell}z^s(1-z)^{\ell-s}\textcolor{black}{{\eta}_{s}^{(i)}}
.
\end{align}
Fix an erasure pattern with $s$ erasures. 
To simplify notation in what follows, let us assume 
that the $s$ erasures are in the last $s$ positions. As in \eq{Wi-def}, 
let us write $\uuu = (\vvv,u_i,{\uuu}')$ for the vector encoded by the 
polar transformation in \eq{polar-def}, and let 
$\yyy = (y_1,y_2,\ldots,y_{\ell-s})$ 
denote the vector observed at the channel output \textcolor{black}{in which the erasure locations are removed}. Then
$
\yyy = \uuu K{|}_{\ell-s}
$
where $K{|}_{\ell-s}$ denotes the submatrix of $K$ consisting of
its first $\ell-s$ columns. Notice, however, that in addition to $\yyy$,
the successive-cancellation decoder knows the vector $\vvv$ consisting
of the first $i-1$ bits of $\uuu$. Thus let us write~
\begin{align}
\yyy 
\,=\, 
\uuu K{\bigm|}_{\ell-s}
=\, 
(\vvv,u_i,{\uuu}') K{\bigm|}_{\ell-s}
=\, 
(\vvv,0,\zero) K{\bigm|}_{\ell-s} \,+\,~
(\zero,u_i,\uuu') K{\bigm|}_{\ell-s} 
\end{align}
and define
$
\xxx = \yyy - (\vvv,0,\zero) K{|}_{\ell-s}
$.
Since this vector $\xxx$ can be computed by the decoder, it
follows that~the decoding task is to determine $u_i$ given
\begin{align}
\xxx 
\: = \,
(u_i,\uuu') K'{\bigm|}_{\ell-s} 
\end{align}
where $K'{|}_{\ell-s}$ denotes the submatrix of $K{|}_{\ell-s}$ consisting of
its last $\ell-(i-1)$ columns. It is easy to see that~$u_i$ can be determined
uniquely from $\xxx$ \textcolor{black}{if and only if}
the vector $(1,0,0,\ldots,0)^t$ is in 
the column space of the matrix $K'{|}_{\ell-s}$. Thus we arrive
at the following decodability condition:
\begin{align}
\label{eq:BEC_subkernel_step3}
u_i~ \text{is decodable} 
\:~\Longleftrightarrow~
(1,\underbrace{0,0,\dots,0}_{\ell-i})^t \in\, 
\text{column space of } K'{\bigm|}_{\ell-s}
\vspace*{-0.54ex}
\end{align}
As we shall see, it is advantageous to rephrase this condition in terms
of the column space of the $\ell \times (\ell-s)$ matrix $K{|}_{\ell-s}$,
since we know that all the columns of this matrix are linearly independent.
Clearly, $(1,0,0,\ldots,0)^t$ is in the column space of $K'{|}_{\ell-s}$ 
if and only if
\begin{align}\label{eq:BEC_subkernel_step4}
\exists\, \psi_1, \psi_2, \ldots, \psi_{i-1} \in \Ftwo: 
\hspace{3ex}
(\psi_1,\psi_2\cdots,\psi_{i-1},1,\underbrace{0,0,\cdots,0}_{\ell-i})^t 
\in ~\text{column space of } K{\bigm|}_{\ell-s}. 
\vspace{-0.54ex}
\end{align}
Now let $\textcolor{black}{{\bf e}}_j$ denote the $j$-th element of the canonical basis
for $\Ftwo^\ell$ and define the linear subspace $E_j$ of $\mathbb{F}_2^{\ell}$ 
as 
\smash{$E_j \triangleq \big\langle \textcolor{black}{{\bf e}}_1,\textcolor{black}{{\bf e}}_2,\cdots,\textcolor{black}{{\bf e}}_j \big\rangle$},
where $\langle \cdot \rangle$ denotes the linear span over $\Ftwo$.
With this, in view of \eq{eq:BEC_subkernel_step3} and 
\eq{eq:BEC_subkernel_step4}, the~decodability condition 
can be rephrased as follows:
\begin{align}
\label{eq:prelim_decodability_condition}
u_i~ \text{is decodable} 
\:~\Longleftrightarrow~
(E_i\setminus E_{i-1}) \cap \bigl(\text{column space of } K{\bigm|}_{\ell-s}\bigr) 
\:\neq\: \emptyset.
\end{align}
In what follows, we use \eqref{eq:prelim_decodability_condition} to
derive an explicit formula for the probability that $u_i$ is decodable
--- that is, for 
$
\mathbb{P}\bigl\{
(E_i\setminus E_{i-1}) \cap \bigl(\text{column space of } K{|}_{\ell-s}\bigr)
\neq \emptyset
\bigr\}
$
when $K$ is selected uniformly at random from $\GL(\ell,\Ftwo)$.

\subsection{Average polarization behavior}\label{subsec:random_kernel_analysis}

In this subsection, we study the erasure probability
of the $i$-th bit-channel $W_i$ given that (\emph{i}) the kernel is selected
uniformly at random from $\GL(\ell,\Ftwo)$, 
and (\emph{ii}) the transmission channel is BEC($z$).
Explicitly, for all $i\in [\ell]$,~we\linebreak
define the \emph{average erasure probability} 
$\cF_i(z)$ as follows:\vspace{0.90ex}
\begin{align}\label{eq:definition_randomkernel}
\cF_i(z) 
~\triangleq~
\mathbb{E}_K \bigl[f_{K,i}(z)\bigr] 
~=~
\frac{\displaystyle\sum_{K\in\GL(\ell,\Ftwo)}\hspace{-2.70ex}f_{K,i}(z)}
{\displaystyle\Strut{2.70ex}{0ex}\bigl|\GL(\ell,\Ftwo)\bigr|}
~=~
\frac{\displaystyle\sum_{K\in\GL(\ell,\Ftwo)}\hspace{-2.70ex}f_{K,i}(z)}
{\displaystyle \prod_{j=0}^{\ell-1} (2^\ell - 2^j)}
~.
\end{align}
In what follows, we analyze the asymptotic behavior of $\cF_i(z)$ and
show that, as $\ell$ grows, $\cF_i(z)$ becomes close to a step
function with a jump at $z\sim \sfrac{i}{\ell}$. \textcolor{black}{Later on, we} prove concentration results \textcolor{black}{that} show that, with high probability over the choice of the kernel, $f_{K,i}(z)$ is also close to a sharp step
function centered around $z\sim\sfrac{i}{\ell}$. \textcolor{black}{This is captured in Propositions~\ref{lemma:randomkernel_caseI_erasure_probability_lowerbound} and~\ref{lemma:randomkernel_caseI_erasure_probability_upperbound}.}

\textcolor{black}{
\begin{proposition}[Lower bound on the average erasure probability]
	\label{lemma:randomkernel_caseI_erasure_probability_lowerbound}
	Let $\cF_i(z)$ be the average erasure probability of the $i$-th
	bit-channel as defined in~\eqref{eq:definition_randomkernel}. 
	Fix $\beta,\sigma\in\mathbb{R}^+\triangleq\{x:x\in\R, x>0\}$ and assume that
	\begin{equation}
	\label{eq:hpz}
	z 
	~>~ 
	\frac{i}{\ell} + \frac{\ceil{\sigma\log \ell}}{\ell}+\bigg(\frac{\beta\ln \ell}{2\ell}\bigg)^{1/2},
	\end{equation}
	where $\log$ and\, $\ln$ denote the logarithm in base $2$ and $e$, respectively. 
	Then, we have that
	\begin{align}\label{eq:randomkernel_caseI_erasure_probability_lowerbound}
	\cF_i(z) ~>~ \bigl(1-\ell^{-\beta}\bigr)\bigl(1-\ell^{-\sigma}\bigr).
	\end{align}
\end{proposition}
\begin{proposition}[Upper bound on the average erasure probability]
	\label{lemma:randomkernel_caseI_erasure_probability_upperbound}
	Let $\cF_i(z)$ be the average erasure probability of the $i$-th bit-channel as defined in~\eqref{eq:definition_randomkernel}. 
	Fix $\beta,\sigma\in\mathbb{R}^+$ and assume that
	\begin{equation}\label{eq:hpz2}
	z
	~<~
	\frac{i}{\ell} - \frac{\textcolor{black}{h}(\sigma)}{\ell}
	\,-\,
	\bigg(\frac{\beta\ln \ell}{2\ell}\bigg)^{1/2},
	\end{equation}
	where $\log$ and\, $\ln$ denote the logarithms in base $2$ and $e$, 
	respectively, and 
	\begin{align}
	\label{eq:randomkernel_caseI_erasure_probability_upperbound_g(delta)}
	\textcolor{black}{h}(\sigma) 
	~=\,
	\left\lfloor \frac{\sigma\log \ell + \log 6}{\log 3 -1} \right\rfloor 
	\,=~
	O(\sigma\log \ell).
	\end{align}
	Then, we have that
	\begin{align}\label{eq:randomkernel_caseI_erasure_probability_upperbound}
	\cF_i(z) < \ell^{-\beta}+\ell^{-\sigma}.
	\end{align}
\end{proposition}
}

\textcolor{black}{Recall} that $\cF_i(z)$ is the probability of observing an erasure
at the $i$-th bit-channel, when there are two sources of randomness:
(\emph{i}) the selection of the kernel, and (\emph{ii}) the number 
and location of the erased bits. Let the random variable $S$
denote the number of erased bits at the receiver. As $z$ is the
erasure probability of the underlying transmission channel, we 
have that
\begin{align}
\label{eq:randomkernel_definition_S}
\P\{ S = s \} ~=~ \binom{\ell}{s}z^s(1-z)^{\ell-s}. 
\end{align}
\looseness=-1
Since we also average over all $\ell\times\ell$ nonsingular kernels,
the location of these $s$ erasures does not affect the average erasure
probability. 
Hence, without loss of generality, we can assume that the erasures 
are in the last $s$ positions. Let
$\mathcal{R}_{\ell-s}\subset \mathbb{F}_2^{\ell}$ denote the linear
span of the first $\ell-s$ columns of the kernel. Since the kernel is
selected uniformly at random from $\GL(\ell,\Ftwo)$, it is easy to see that
$\mathcal{R}_{\ell-s}$ is also chosen uniformly at random from all
subspaces of dimension $\ell-s$ in $\mathbb{F}_2^{\ell}$. Recalling
the decodability condition
(\ref{eq:prelim_decodability_condition}), we have that
\begin{align}
\label{eq:randomkernel_bitchannel_erasureprobability_simplified}
\mathbb{P} \{u_i = \text{erasure} \,|\, S=s \} 
~=~
\P \bigl\{ \mathcal{R}_{\ell-s} \cap (E_i\setminus E_{i-1}) = \emptyset\bigr\},
\end{align}
where $\mathcal{R}_{\ell-s}$ is a subspace of dimension $\ell-s$ in $\mathbb{F}_2^{\ell}$ that is chosen uniformly at random. Note that the \emph{event} on the Left Hand Side (LHS) is reliant on a specific number of erasures, $s$, \textcolor{black}{and} is computed over all possible \textcolor{black}{locations of erasures} and selections of $K$. However, the \emph{event} on the RHS is independent of the location and number of erasures, and thus is computed over all selections of random subspace $\mathcal{R}_{\ell-s}$. Therefore, the probability that \textcolor{black}{the} $i$-th bit is erased given $s$ erasures is a claim solely on the structure of the kernel. Now, we can rewrite $\cF_i(z)$ as
\begin{align}
\label{eq:F_i_formula}
\cF_i(z) 
~=~
\sum_{s=0}^{\ell} \P\{S = s\}\mathbb{P} \{u_i = \text{erasure} \,|\,S=s \}
~=~
\sum_{s=0}^{\ell} \binom{\ell}{s}z^s(1-z)^{\ell-s}p_{i\mid s}\hspace{1ex},
\end{align}
where we define the \emph{average conditional erasure probability} $p_{i|s}$ 
as follows:
\begin{align}\label{eq:defacep}
p_{i|s} 
~\triangleq~ 
\mathbb{P} \{u_i = \text{erasure} \,|\, S=s \}\hspace{1ex} 
\textcolor{black}{=}
\hspace{1.2ex}
\P \{ \mathcal{R}_{\ell-s} \cap (E_i\setminus E_{i-1}) = \emptyset\},
\end{align}
\textcolor{black}{where the right-most equality is derived from~(\ref{eq:randomkernel_bitchannel_erasureprobability_simplified}).}
\begin{lemma}[Closed-form for the average conditional erasure probability]
\label{lemma:randomkernel_conditional_erasure_probability}
Let $p_{i|s}$ be the average conditional erasure probability defined in  \eqref{eq:defacep}. Then, for any $i$ and $s$, we have
\begin{align}\label{eq:randomkernel_conditional_erasure_probability}
p_{i|s}
~=~
\G{\ell}{\ell-s}^{-1}\hspace{2ex}\sum_{t=\max\{i-s,0\}}^{\min \{\ell-s ,i-1\}}
\G{i-1}{t} \prod_{j=0}^{\ell-s-t-1}\frac{2^{\ell}-2^{i+j}}{2^{\ell-s}-2^{t+j}}\hspace{1ex},
\end{align}
\textcolor{black}{where {\small{$\begin{bmatrix}a \\ b\end{bmatrix}$}} is the binary Gaussian binomial coefficient that denotes the total number of subspaces with \linebreak dimension $b$ in $\mathbb{F}_2^{a}$.}
\end{lemma}

\begin{proof}
\label{lemma:randomkernel_conditional_erasure_probability_extended}
Let $\Delta_{\ell-s}$ denote the number of subspaces of dimension $\ell-s$
in $\mathbb{F}_2^{\ell}$. That is,
\begin{align}\label{eq:randomkernel_conditional_erasure_delta_definition}
\Delta_{\ell-s}
~\triangleq~ 
\G{\ell}{\ell-s}
~=~
\prod_{j=0}^{\ell-s-1}\frac{2^{\ell}-2^j}{2^{\ell-s}-2^j}.
\end{align}
Define $\Gamma(t;\ell,s,i)$ as the number of subspaces $A$ of
dimension $\ell-s$ in $\mathbb{F}_2^{\ell}$ such that $A \cap
(E_i\setminus E_{i-1}) = \emptyset$ and $\dim(A\cap E_{i-1}) =
t$. \textcolor{black}{Recall that $E_i$ and $E_{i-1}$ are linear subspaces of $\Ftwo^{\ell}$ with respective dimensions of $i$ and $i-1$, and $E_{i-1}\subset E_i$. Therefore, $\Gamma(t;\ell,s,i)$ is equal to} the number of subspaces $A$ of dimension $\ell-s$ in $\mathbb{F}_2^{\ell}$ such~that
$\dim(A\cap E_{i-1}) = \dim(A\cap E_{i}) = t$. Consequently, the
integer $t$ in the definition of $\Gamma(t;\ell,s,i)$ satisfies~
\begin{align}
\label{eq:randomkernel_conditional_erasure_t_definition}
\max\{i-s,0\}  \,\leq\,  t \,\leq\, \min \{\ell-s ,i-1\}.
\end{align}
A simple basis counting argument 
(see, for example, \cite[Section\,II.C]{Vardy94}) yields that 
\begin{align}
\label{eq:randomkernel_conditional_erasure_gamma_definition}
\Gamma(t;\ell,s,i)
~=~ 
\G{i-1}{t} \prod_{j=0}^{\ell-s-t-1}\frac{2^{\ell}-2^{i+j}}{2^{\ell-s}-2^{t+j}}
\end{align}
where the first term in \eq{eq:randomkernel_conditional_erasure_gamma_definition}
counts the number of subspace of dimension $t$ in $E_{i-1}$ whereas
the second term counts the (normalized) number of basis extensions
from dimension $t$ to dimension $\ell-s$.
Enumerating over~all possible values of $t$ given by 
\eq{eq:randomkernel_conditional_erasure_t_definition},
the desired conditional erasure probability can be written as
\begin{align}\label{eq:randomkernel_conditional_erasure_probability_parameters}
p_{i|s} ~=~ \frac{\displaystyle\sum_{t=\max\{i-s,0\}}^{\min\{\ell-s,i-1\}}
\hspace{-2.10ex}\Gamma(t;\ell,s,i)}{\Delta_{\ell-s}}
~=~
\G{\ell}{\ell-s}^{-1}\hspace{2ex}\sum_{t=\max\{i-s,0\}}^{\min \{\ell-s ,i-1\}}
\G{i-1}{t} \prod_{j=0}^{\ell-s-t-1}\frac{2^{\ell}-2^{i+j}}{2^{\ell-s}-2^{t+j}}
~.
\end{align}
\end{proof}

Next, we use this closed-form expression to provide upper and lower 
bounds on the average conditional erasure probability $p_{i|s}$ and 
on the average erasure probability $\cF_i(z)$.

\begin{lemma}[Lower bound on the average conditional erasure probability]
\label{lemma:randomkernel_caseI_conditional_erasure_probability_lowerbound}
Let $p_{i|s}$ be the average conditional erasure probability defined 
in \eqref{eq:defacep}. Then, for any $i$ and $s$, we have
\begin{align}\label{eq:randomkernel_caseI_erasure_probability_lowerbound_step1}
p_{i|s} ~\geq~ 1-2^{-(s-i)}.
\end{align}
\end{lemma}

\begin{proof}
If $i\ge s$, then the lemma holds vacuously. 
Henceforth, let us assume that $i< s$. 
We drop all but the first term 
from~(\ref{eq:randomkernel_conditional_erasure_probability_parameters}) to write
\begin{align}
p_{i|s} 
~=~
\frac{\displaystyle
\sum_{t=0}^{\min\{\ell-s,i-1\}}\hspace{-2.70ex}\Gamma(t,\ell,s,i)}{\Delta_{\ell-s}}
~\geq~ \frac{\Gamma(0;\ell,s,i)}{\Delta_{\ell-s}}
~=~
\frac
{\displaystyle\prod_{j=0}^{\ell-s-1}\frac{2^\ell - 2^{i+j}}{2^{\ell-s}-2^j}}
{\displaystyle\prod_{j=0}^{\ell-s-1}\frac{2^\ell - 2^{j}}{2^{\ell-s}-2^j}} 
~=~
\prod_{j=0}^{\ell-s-1}\frac{2^\ell - 2^{i+j}}{2^{\ell}-2^{j}}.
\end{align}
The proof now reduces to the following calculation:
\begin{align}
\prod_{j=0}^{\ell-s-1}\frac{2^\ell - 2^{i+j}}{2^{\ell}-2^{j}} > \prod_{j=0}^{\ell-s-1}\frac{2^\ell - 2^{i+j}}{2^{\ell}} = \prod_{j=0}^{\ell-s-1}\left( 1-2^{-(\ell-i) +j} \right)
\textcolor{black}{\stackrel{\mathclap{\mbox{\footnotesize(a)}}}{\geq}}
1 - \sum_{j=0}^{\ell-s-1} 2^{-(\ell-i)+j} 
\textcolor{black}{\stackrel{\mathclap{\mbox{\footnotesize(b)}}}{>}}
 1-2^{-(s-i)},
\end{align}
\textcolor{black}{where (a) is because of 
\begin{align}
\prod_{i=1}^n (1-x_i) \geq 1-\sum_{i=1}^n x_i \hspace{1.5ex}\text{ for all }\hspace{1.5ex} \textcolor{black}{0\leq x_i\leq 1},
\end{align}
\textcolor{black}{which can be shown by induction}, and (b) is due to the fact that for all $n\in\N$, we have $2^{-n} = \sum_{i=n+1}^{\infty} 2^{-i}$. }
\end{proof}

\begin{proof}[Proof of Proposition~\ref{lemma:randomkernel_caseI_erasure_probability_lowerbound}]
We begin by dropping the first \textcolor{black}{$i+\ceil{\textcolor{black}{\sigma} \log (\ell)}+1$ terms} in~(\ref{eq:F_i_formula}) and applying 
\Lref{lemma:randomkernel_caseI_conditional_erasure_probability_lowerbound}
to obtain
\begin{align}
\label{eq:randomkernel_caseI_erasure_probability_lowerbound_step2}
\begin{split}
\cF_i(z) 
~=~ \sum_{s=0}^{\ell} \binom{\ell}{s}z^s(1-z)^{\ell-s}p_{i|s}\hspace{2ex}&> \sum_{s=i+\ceil{\textcolor{black}{\sigma}\log \ell}+1}^{\ell} \binom{\ell}{s}z^s(1-z)^{\ell-s}(1-2^{-(s-i)})\\
&\geq (1-\ell^{-\textcolor{black}{\sigma}})\sum_{s=i+\textcolor{black}{\sigma}\ceil{\log \ell}+1}^{\ell} \binom{\ell}{s}z^s(1-z)^{\ell-s}.
\end{split}
\end{align}
Now, we point out that the sum on the RHS of
\eqref{eq:randomkernel_caseI_erasure_probability_lowerbound_step2} is
the tail probability of a binomial distribution with $\ell$ trials and
a success rate of $z$. More formally, \textcolor{black}{$X\sim B(\ell,z)$ is a binomial random variable with $\ell$ trials and success probability $z$}. Then, from
\eqref{eq:randomkernel_caseI_erasure_probability_lowerbound_step2} we
immediately obtain that
\begin{align}
\cF_i(z) 
\,>\, 
\bigl(1-\ell^{-\textcolor{black}{\sigma}}\bigr)\, \P \{X> i + \ceil{\textcolor{black}{\sigma}\log \ell} \}.
\end{align}
\textcolor{black}{Now, we invoke Hoeffding's inequality~\cite{Hoef63} for a sequence of of i.i.d. Bernoulli random variables with success probability $1-z$ and $\ell$ trials, which states that for any $\nu>0$, 
\begin{align}
\mathbb{P}\big(X\geq (1-z+\nu)\ell\big) \leq \exp (-2\nu^2\ell).
\end{align} 
By replacing the value of $\nu$ with the expressions from~(\ref{lemma:randomkernel_caseI_erasure_probability_lowerbound}), we have}
\begin{align}\label{eq:randomkernel_caseI_erasure_probability_lowerbound_step3}
\begin{split}
\mathbb{P} \{X> i + \ceil{\textcolor{black}{\sigma}\log \ell}\} 
&= 1- \mathbb{P} \{X\leq i + \ceil{\textcolor{black}{\sigma}\log \ell}\}\\[0.90ex]
&\stackrel{\mathclap{\mbox{\footnotesize(a)}}}{\ge} 1- \exp\left(-2\frac{\big(z\ell - (i+\ceil{\textcolor{black}{\sigma}\log \ell})\big)^2}{\ell}\right)\\
&\stackrel{\mathclap{\mbox{\footnotesize(b)}}}{\ge} 1-\ell^{-\beta},
\end{split}
\end{align}
where in (a) we have used Hoeffding's inequality and 
in (b) we have used \eqref{eq:hpz}. The lemma now readily~follows
by combining
\eqref{eq:randomkernel_caseI_erasure_probability_lowerbound_step2}
and
\eqref{eq:randomkernel_caseI_erasure_probability_lowerbound_step3}.
\end{proof}

Next, we use the closed-form expression in 
\Lref{lemma:randomkernel_conditional_erasure_probability}
in order to derive a lower bound on the average conditional 
erasure probability and on the average erasure probability.

\begin{lemma}[Upper bound on the average conditional erasure probability]
\label{lemma:randomkernel_caseI_conditional_erasure_probability_upperbound}
Let $p_{i|s}$ be the average conditional erasure probability defined in  \eqref{eq:defacep}. Then, for any $i$ and $s$,
\begin{align}\label{eq:randomkernel_caseI_erasure_probability_upperbound_step3}
p_{i|s} ~\leq~ 2 \left( \frac{2}{3}\right)^{i-s-1}.
\end{align}
\end{lemma}

\begin{proof}
If $s\ge i-1$, the bound holds vacuously. 
Henceforth, let us assume that $s< i-1$. We start by proving that the
term with $t= i-s$ is the dominant one in the expression
\eqref{eq:randomkernel_conditional_erasure_probability_parameters} for
$p_{i|s}$. For all $t> i-s$, we have that
\begin{align}
\frac{\Gamma(t;\ell,s,i)}{\Gamma(t-1;\ell,s,i)} 
&= \frac{\begin{bmatrix}i-1\\t\end{bmatrix}}{\begin{bmatrix}i-1\\t-1\end{bmatrix}}\times
\bigg({\prod_{j=0}^{\ell-s-t-1} \frac{2^\ell -2^{i+j} }{2^{\ell-s}-2^{t+j}}}\bigg)
\bigg/
\bigg({\prod_{j=0}^{\ell-s-t} \frac{2^\ell - 2^{i+j}}{2^{\ell-s}-2^{t+j-1}}}\bigg),
\end{align}
which using a straightforward manipulation can be simplified as
\begin{align}
\frac{(2^{i-1}-2^{t-1})(2^{\ell-s}-2^{t-1})}{2^{t-1}(2^t-1)(2^{\ell}-2^{i+\ell-s-t})}
\leq \frac{1}{2^{t-1}} \cdot \frac{2^{i-1}\cdot 2^{\ell-s}}{2^{t-1}\cdot 2^{\ell-1}}
=\frac{ 2^{i-s-t+1}}{2^{t-1}} \leq 2^{-t+1}\leq \frac{1}{2}.
\end{align}
Therefore, for any $t> i-s$, we have that
\begin{align}
\Gamma(t;\ell,s,i) 
~\leq~ 
2^{-(t-(i-s))}\hspace{0.5ex}\Gamma(i-s;\ell,s,i),
\end{align}
which implies that
\begin{align}
\label{eq:randomkernel_caseI_erasure_probability_upperbound_step1_c}
p_{\textcolor{black}{i|s}} 
~\leq~
\frac{\Gamma(i-s;\ell,s,i)}{\Delta_{\ell-s}}
\Big(1 + 2^{-1} + 2^{-2} + \cdots \Big) 
~\leq~
\frac{2\Gamma(i-s;\ell,s,i)}{\Delta_{\ell-s}}.
\end{align}
In a similar fashion, we fix $\ell$ and $i$, and study the
exponential decay of the dominant term in $p_{i|s}$, denoted by 
$\xi_s \triangleq \Gamma(i-s;\ell,s,i)/\Delta_{\ell-s}$, 
as $s$ decreases. We again use straightforward manipulation to obtain
\begin{align}\label{eq:randomkernel_caseI_erasure_probability_upperbound_step2_b}
\begin{split}
\frac{\xi_{s}}{\xi_{s+1}} 
&= \frac{\Delta_{\ell-s-1}}{\Delta_{\ell-s}}\times\frac{\begin{bmatrix}i-1\\i-s\end{bmatrix}}{\begin{bmatrix}i-1\\i-s-1\end{bmatrix}}\times
\bigg(\prod_{j=0}^{\ell-i-1} \frac{2^{\ell}-2^{i+j}}{2^{\ell-s}-2^{i-s+j}}  \bigg)
\bigg/
\bigg(\prod_{j=0}^{\ell-i-1} \frac{2^{\ell}-2^{i+j}}{2^{\ell-s-1}-2^{i-s-1+j}}  \bigg)\\
&=\frac{(2^{i-1}-2^{i-s-1})(2^{\ell-s}-1)}{(2^{i-s}-1)(2^{\ell}-2^{\ell-s-1})} = \bigg(\frac{2^s-1}{2^{s+1}-1}\bigg)\frac{1-2^{-(\ell-s)}}{1-2^{-(i-s)}}\\
&\leq \frac{1}{2} \times \frac{1}{1-2^{-(i-s)}}\leq \frac{1}{2}\times\frac{1}{1-1/4} = \frac{2}{3}.
\end{split}
\end{align}
As a result, we conclude that, for any $s<i-1$,
\begin{align}\label{eq:randomkernel_caseI_erasure_probability_upperbound_step2_c}
p_{i|s}
\hspace{1ex} \overset{(\ref{eq:randomkernel_caseI_erasure_probability_upperbound_step1_c})}{\leq}\hspace{1ex} 
\frac{2\Gamma(i-s;\ell,s,i)}{\Delta_{\ell-s}} 
\hspace{1ex}\overset{(\ref{eq:randomkernel_caseI_erasure_probability_upperbound_step2_b})}{\leq}\hspace{1ex}
2\xi_{i-1} \left(\frac{2}{3}\right)^{i-s-1} \leq 2\left(\frac{2}{3}\right)^{i-s-1},
\end{align} 
\textcolor{black}{
where the last inequality follows from the fact that $\Gamma(i-s;\ell,s,i)$ is the number of subspaces of dimension $\ell-s$ in $\Ftwo^{\ell}$ with some additional properties, while $\Delta_{\ell-s}$ denotes the total nummber of such subspaces, and thus, $\xi_s \leq 1$ for all $s$ that $\xi_s$ is well defined. 
}
\end{proof}

\begin{proof}[Proof of Proposition~\ref{lemma:randomkernel_caseI_erasure_probability_upperbound}]
Let us recall the formulation of $\cF_i(z)$ from~(\ref{eq:F_i_formula}) and split the summation into two parts, where a trivial upper bound is applied to each part: we drop $\binom{\ell}{s}z^s(1-z)^{\ell-s}$ for all terms in the summation with $s\leq i-\textcolor{black}{h}(\textcolor{black}{\sigma}) -1$, and we drop $p_{i|s}$ from the remaining terms that correspond to $s\geq i-\textcolor{black}{h}(\textcolor{black}{\sigma})$. More formally, we have
\begin{align}\label{eq:randomkernel_caseI_erasure_probability_upperbound_step3_a}
\begin{split}
\cF_i(z) &= \sum_{s=0}^{i-\textcolor{black}{h}(\textcolor{black}{\sigma})-1} \binom{\ell}{s}z^s(1-z)^{\ell-s} p_{i|s} \hspace{1ex}+ \sum_{s=i-\textcolor{black}{h}(\textcolor{black}{\sigma})}^{\ell} \binom{\ell}{s}z^s(1-z)^{\ell-s} p_{\textcolor{black}{i|s}}
\\
&< \sum_{s=0}^{i-\textcolor{black}{h}(\textcolor{black}{\sigma})-1} p_{\textcolor{black}{i|s}} \hspace{16.8ex}+ \sum_{s=i-\textcolor{black}{h}(\textcolor{black}{\sigma})}^{\ell} \binom{\ell}{s}z^s(1-z)^{\ell-s}.
\end{split}
\end{align}
We apply the upper bound in~(\ref{eq:randomkernel_caseI_erasure_probability_upperbound_step3}) to the first summation, and obtain that
\begin{align}\label{eq:randomkernel_caseI_erasure_probability_upperbound_step3_b}
\sum_{s=0}^{i-\textcolor{black}{h}(\textcolor{black}{\sigma})-1} p_{\textcolor{black}{i|s}} \leq \sum_{s=0}^{i-\textcolor{black}{h}(\textcolor{black}{\sigma})-1} 2\left(\frac{2}{3}\right)^{i-s-1} \leq \sum_{s=\textcolor{black}{h}(\textcolor{black}{\sigma})}^{\infty} 2\left(\frac{2}{3}\right)^s = 
 6\left(\frac{2}{3}\right)^{\textcolor{black}{h}(\textcolor{black}{\sigma})} = \ell^{-\textcolor{black}{\sigma}}.
\end{align}
\textcolor{black}{Utilizing the assumption in~(\ref{eq:hpz2}),} the second summation is again upper bounded by applying Hoeffding's inequality on the tail probability of the binomial distribution $X\sim B(\ell,z)$ with $\ell$ trials and a success rate of $z$ as follows:
\begin{align}\label{eq:randomkernel_caseI_erasure_probability_upperbound_step3_c}
\sum_{s=i-\textcolor{black}{h}(\textcolor{black}{\sigma})}^{\ell} \binom{\ell}{s}z^s(1-z)^{\ell-s} = \mathbb{P} \{X\geq i - \textcolor{black}{h}(\textcolor{black}{\sigma})\} \leq \mathbb{P} \left\{X\geq z\ell + \left(\frac{\beta\ell\ln \ell}{2} \right)^{1/2} \right\} \leq \ell^{-\beta}. 
\end{align}
\end{proof}


\subsection{Proof of~\Pref{thm:cases}}\label{subsec:proof_thm3}
At this point, we have gathered all the required tools to prove~\Pref{thm:cases}. Our proof consists of two steps. First, we show that the polarization behavior of a random non-singular $\ell\times\ell$ kernel is given, with high probability, by the function $\cF_i(z)$ analyzed in the previous subsection. Then, we explain how to relate this fact to an upper bound on $\lambda_{\alpha,K}(z)$. \textcolor{black}{Note that throughout this section, all probabilities are defined with respect to the random selection of non-singular kernels and there is no randomness in $i$. In fact, the polarization behavior of a desired kernel should be similar to $\cF_i(z)$ for all $i$.}
As the theorem suggests, we split the proof into two parts: the first
part takes care of the middle interval and proves
\eqref{eq:cases_middle}, while the second part takes care of the tail
intervals and proves \eqref{eq:cases_tails}.

\begin{proof}[Proof of \eqref{eq:cases_middle}]
First, we combine the results of 
\Pref{lemma:randomkernel_caseI_erasure_probability_lowerbound}
and 
\Pref{lemma:randomkernel_caseI_erasure_probability_upperbound} 
to show that $\cF_i(z)$ roughly behaves as a step function. 
In the previous subsection, we have shown that
\begin{align}
\label{eq72}
\left\{
   \begin{array}{ll}
      \cF_i(z) >(1-\ell^{-\beta})(1-\ell^{-\textcolor{black}{\sigma}}),  &\hspace{7ex}\text{ if }  z>\frac{i}{\ell} + \frac{\ceil{\textcolor{black}{\sigma}\log \ell}}{\ell}+
               \bigg(\frac{\beta\ln \ell}{2\ell}\bigg)^{1/2} \\
      \cF_i(z) <\ell^{-\beta}+\ell^{-\textcolor{black}{\sigma}},  &\hspace{7ex}\text{ if }  z<\frac{i}{\ell} - \frac{\floor{\frac{\textcolor{black}{\sigma}\log \ell + \log 6}{\log 3 -1}}}                              {\ell}-\bigg(\frac{\beta\ln \ell}{2\ell}\bigg)^{1/2}
   \end{array}.
\right.
\end{align}
Our strategy is to show that, with high probability over the choice of
the kernel, $f_{K,i}(z)$ is sharp for each fixed value of $i$. Then,
we will use a \textcolor{black}{union-bound-like argument} to show that $f_{K,i}(z)$ is sharp for all
$i\in[\ell]$. To this end, we first set 
$\beta = \textcolor{black}{\sigma} = 4.5+\log\ell$ in \eq{eq72}. 
\textcolor{black}{Given that the nature of our results is asymptotic,} we assume that $\ell \geq 32$. Now, it is easy to derive the following from \eq{eq72}. 
\begin{align}\label{eq:randomkernel_caseI_erasure_probability_final_bounds}
\left\{
   \begin{array}{ll}
      \cF_i(z) >1-2\ell^{-4.5-\log \ell}>1-(2\ell^{4+\log \ell})^{-1},  &\hspace{7ex}\text{ if }  z\geq\frac{i}{\ell} + c_3\ell^{-1/2}\log \ell\\
      \\
      \cF_i(z) <2\ell^{-4.5-\log\ell}<(2\ell^{4+\log \ell})^{-1},  &\hspace{7ex}\text{ if }  z\leq\frac{i}{\ell} - c_4\ell^{-1/2}\log \ell
   \end{array},
\right.
\end{align}
where
\textcolor{black}{
\begin{align}\label{eq:randomkernel_caseI_erasure_probability_final_bounds_c_bound}
\begin{split}
c_3 &= \max_{\ell\geq 32} 
\frac
{\frac{\ceil{(4.5+\log \ell)\log \ell}}{\ell} + \big(\frac{(4.5+\log \ell)\ln \ell)}{2\ell}\big)^{1/2}}
{\ell^{-1/2}\log \ell}, 
\\[2ex]
c_4 &= \max_{\ell\geq 32} 
\frac
{\frac{\ceil{\frac{(4.5+\log\ell)\log\ell + \log 6}{\log 3 -1}}}{\ell} + \big(\frac{(4.5+\log \ell)\ln \ell)}{2\ell}\big)^{1/2}}
{\ell^{-1/2}\log \ell}.
\end{split}
\end{align}}
\textcolor{black}{Note that both $c_3$ and $c_4$ are finite numbers since the numerators in the RHSs of~(\ref{eq:randomkernel_caseI_erasure_probability_final_bounds_c_bound}) decay faster than their respective denominators as $\ell$ grows. It is also possible to remove the ceilings and show that both expressions are decreasing functions of $\ell$ if $\ell \geq 32$, which means that they attain their maximum at $\ell =32$. Thus, $c_3< 2.51$ and $c_4< 3.86$.
}

\textcolor{black}{Our goal is to prove the simultaneous concentration of $f_{K,i}(z)$'s around their means, $\mathcal{F}_i(z)$, with regards to where and how fast they transition from $f_{K,i}(z)\sim0$ to $f_{K,i}(z)\sim1$.} 
To do so, we first show that for any fixed value of $i$, $f_{K,i}(z)$ behaves similar to the average behavior, with high probability over the choice of $K$. Next, we provide a union-bound-like argument to prove that, with high probability over the choice of $K$, $f_{K,i}(z)$ is close to the average for all values of $i$. \textcolor{black}{For the first step, we recall the erasure probability of the $i$-th bit-channel from~(\ref{def:f_K,i(z)}) and expand it as}
\begin{align}
f_{K,i}(z) = \sum_{\text{all erasure patterns }{\bf e}\in\Ftwo^{\ell}} \one[\text{${\bf e}$ makes $i$-th bit-channel undecodable}] z^{wt({\bf e})}(1-z)^{\ell-wt({\bf e})}.
\end{align}
\textcolor{black}{Considering that each $z^{wt({\bf e})}(1-z)^{\ell-wt({\bf e})}$ term in the function above is a continuous and increasing function of $z$, we deduce that $f_{K,i}(z)$ is also a continuous and increasing function of $z$}. Therefore, to show the sharp transition of $f_{K,i}(z)$ around $z = \sfrac{i}{\ell}$, it suffices to consider only two points in $(0,1)$, one slightly larger than $z = \sfrac{i}{\ell}$ and one slightly smaller. Let us do so by defining \textcolor{black}{$c_5\triangleq \max\{c_3,c_4\}$} and 
\begin{align}
a_i\triangleq \frac{i}{\ell} + \textcolor{black}{c_5}\ell^{-1/2}\log \ell.
\end{align} From~(\ref{eq:randomkernel_caseI_erasure_probability_final_bounds}) \textcolor{black}{and~(\ref{eq:definition_randomkernel})}, we have that
\begin{align}
\mathbb{E}_{\textcolor{black}{K}}\big[1 - f_{K,i}(a_i)\big] = 1 - \cF_i(a_i) < (2\ell^{4+\log \ell})^{-1}.
\end{align}
From \textcolor{black}{Markov's} inequality, we deduce that
\begin{align}\label{eq:caseI_markov_inequality_upperbound}
\mathbb{P}\big\{f_{K,i}(a_i) \leq 1-\frac{1}{\ell^{2+\log \ell}} \big\} = \mathbb{P}\big\{1 - f_{K,i}(a_i) \geq \frac{1}{\ell^{2+\log \ell}} \big\} \leq \frac{\mathbb{E}_{K}\big[1 - f_{K,i}(a_i)\big]}{1/\ell^{2+\log \ell}} \leq \frac{1}{2\ell^2}\hspace{1ex}.
\end{align}
Define
\begin{align}\label{eq:caseI_upperbound_A_i}
\mathcal{A}_i \triangleq \big\{K\in \mathbb{F}_2^{\ell\times\ell} \big| K \text{ is nonsingular and } f_{K,i}(a_i) \geq 1 - \frac{1}{\ell^{2+\log \ell}}\big\}\hspace{1ex}. 
\end{align}
Therefore, \eqref{eq:caseI_markov_inequality_upperbound} can be re-written as
\begin{align}
\mathbb{P}\{ K \in \mathcal{A}_i \} \geq 1-\frac{1}{2\ell^2}\hspace{1ex}.
\end{align}
Similarly, set 
\begin{align}
b_i\triangleq \frac{i}{\ell} - \textcolor{black}{c_5}\ell^{-1/2}\log \ell,
\end{align}
 and define 
\begin{align}\label{eq:caseI_lowerbound_B_i}
\mathcal{B}_i \triangleq \big\{K\in \mathbb{F}_2^{\ell\times\ell} \big| K \text{ is nonsingular and } f_{K,i}(b_i) \textcolor{black}{\leq} \frac{1}{\ell^{2+\log \ell}}\big\}. 
\end{align}
A very similar use of \textcolor{black}{Markov's} inequality shows that
\begin{align}
\mathbb{P}\{ K \in \mathcal{B}_i \} \geq 1-\frac{1}{2\ell^2}\hspace{1ex}.
\end{align}
Then, define 
\begin{align}
\mathcal{D} = \big(\cap_{j=1}^{\ell} \mathcal{A}_j\big)\cap\big(\cap_{j=1}^{\ell}  \mathcal{B}_j\big).
\end{align}
 By union bound, we obtain that
\begin{align}\label{eq:caseI_union_bound}
\mathbb{P}&\{K \in \mathcal{D}\} 
\geq 1 - \sum_{i=1}^{\ell} \mathbb{P} \{K \notin \mathcal{A}_i\} - \sum_{i=1}^{\ell} \mathbb{P} \{K \notin \mathcal{B}_i\}
\geq 1 -  \frac{2\ell}{2\ell^2} = 1-\frac{1}{\ell}\hspace{1ex}.
\end{align}

We assume that $K\in \mathcal{D}$ throughout the remainder of
proof. This implies that, for $i\in [\ell]$,
\begin{align}\label{eq:caseI_intersection_kernel_step1}
\left\{
   \begin{array}{ll}
      f_{K,i}(z) \textcolor{black}{\geq} 1-\frac{1}{\ell^{2+\log \ell}},  &\hspace{7ex}\text{ for }  z = \frac{i}{\ell} + \textcolor{black}{c_5}\ell^{-1/2}\log \ell\\
      f_{K,i}(z) \textcolor{black}{\leq}\frac{1}{\ell^{2+\log \ell}},  &\hspace{7ex}\text{ for }  z = \frac{i}{\ell} - \textcolor{black}{c_5}\ell^{-1/2}\log \ell
   \end{array}.
\right.
\end{align}
As $f_{K,i}(z)$ is an increasing function of $z$, (\ref{eq:caseI_intersection_kernel_step1}) is equivalent to
\begin{equation}\label{eq:caseI_intersection_kernel_step2}
\left\{
   \begin{array}{ll}
      f_{K,i}(z) \textcolor{black}{\geq}1-\frac{1}{\ell^{2+\log \ell}},  &\hspace{7ex}\text{ for }  z\geq \frac{i}{\ell} + \textcolor{black}{c_5}\ell^{-1/2}\log \ell\\
      f_{K,i}(z) \textcolor{black}{\leq}\frac{1}{\ell^{2+\log \ell}},  &\hspace{7ex}\text{ for }  z\leq \frac{i}{\ell} - \textcolor{black}{c_5}\ell^{-1/2}\log \ell
   \end{array}.
\right.
\end{equation}
Given these concentration results, we can proceed to the second step of the proof. Let us define
\begin{equation}
\begin{split}
T_0(z,\ell) &\triangleq z\ell - \textcolor{black}{c_5}\ell^{1/2}\log \ell, \\
 T_1(z,\ell)& \triangleq z\ell + \textcolor{black}{c_5}\ell^{1/2}\log \ell.
\end{split}
\end{equation}
Note that
\begin{align}\label{eq:caseI_intersection_kernel_step3_number_of_i_s}
\frac{i}{\ell} - \textcolor{black}{c_5}\ell^{-1/2}\log \ell < z <\frac{i}{\ell} + \textcolor{black}{c_5}\ell^{-1/2}\log \ell
\Longleftrightarrow z\ell - \textcolor{black}{c_5}\ell^{1/2}\log \ell < i < z\ell + \textcolor{black}{c_5}\ell^{1/2}\log \ell.
\end{align}
Therefore, for any $z\in \textcolor{black}{(0,1)}$, the number of indices $i$ such that $f_{K,i}(z)$ does not satisfy ~(\ref{eq:caseI_intersection_kernel_step2}) is upper bounded by 
\begin{align}
T_1(z,\ell) - T_0(z,\ell) = 2\textcolor{black}{c_5}\ell^{\textcolor{black}{1/2}}\log \ell.
\end{align}

We can re-write $\lambda_{\alpha,K}(z)$ which was defined earlier in~(\ref{eq:lambda_definition}) as 
\begin{align}\label{eq:caseI_intersection_kernel_step5}
\lambda_{\alpha,K}(z) 
= \frac{\displaystyle\frac{1}{\ell}\sum_{i\in(T_0(z,\ell),T_1(z,\ell))}\hspace{1ex}g_{\alpha}(f_{K,i}(z))}{g_{\alpha}(z)}  + 
\frac{\displaystyle\frac{1}{\ell}\sum_{i\notin(T_0(z,\ell),T_1(z,\ell))}\hspace{1ex}g_{\alpha}(f_{K,i}(z))}{g_{\alpha}(z)} .
\end{align}
By using~(\ref{eq:caseI_intersection_kernel_step2}), we have that, for any $i\notin \big(T_0(z,\ell),T_1(z,\ell)\big)$,
\begin{align}\label{eq:caseI_intersection_kernel_step6}
g_{\alpha}(f_{K,i}(z)) \leq 
g_{\alpha}\left(\frac{1}{\ell^{2+\log \ell}}\right) 
< \big(\ell^{-2-\log \ell}\big)^{\alpha}.
\end{align}
By combining~(\ref{eq:caseI_intersection_kernel_step6}) with the trivial upper bound of $g_{\alpha}(f_{K,i}(z)) \leq 1$ for the left summation, and upper bounding the number of indices not in $(T_0,T_1)$ by $\ell$ \textcolor{black}{for the other summation}, we obtain that
\begin{align}\label{eq:caseI_intersection_kernel_step7}
\lambda_{\alpha,K}(z) \leq \frac{\frac{1}{\ell}2\textcolor{black}{c_5}\ell^{1/2}\log \ell  }{g_{\alpha}(z)} +
 \frac{\big(\ell^{-2-\log \ell}\big)^{\alpha} }{g_{\alpha}(z)}.
\end{align}
Furthermore, note that, for any $z\in (1/\ell^2,1-1/\ell^2)$, 
\begin{align}\label{eq:caseI_intersection_kernel_step8}
g_{\alpha}(z) \geq \big(\ell^{-2}(1-\ell^{-2})\big)^{\alpha}.
\end{align}
By combining~(\ref{eq:caseI_intersection_kernel_step7}) and~(\ref{eq:caseI_intersection_kernel_step8}), we have that
\begin{align}\label{eq:caseI_intersection_kernel_step9}
\lambda_{\alpha,K}(z) \leq \frac{1}{(1-\ell^{-2})^{\alpha}}\bigg( 2\textcolor{black}{c_5}\ell^{-1/2+2\alpha}\log \ell   +\ell^{-\alpha\log \ell} \bigg).
\end{align}
\textcolor{black}
{Given that $\alpha = 1/\log \ell$, we can simplify~(\ref{eq:caseI_intersection_kernel_step7}) according to the following.
\begin{equation}\label{eq:caseI_intersection_kernel_step10a}
\ell^{2\alpha} = (\ell^{\alpha})^2 = 4, \text{ and }
\ell^{-\alpha\log \ell} = \ell^{-1}.
\end{equation}}
Also,  $(1-\ell^{-2})^{-\alpha}$ is \textcolor{black}{a decreasing} function with $\ell$ for $\ell\geq 2$ and \textcolor{black}{an increasing} function with $\alpha$ for $\alpha\leq 1$, which attains its maximum at $(\alpha,\ell) = (1,2)$. That is
\begin{equation}\label{eq:caseI_intersection_kernel_step10b}
\begin{split}
(1-\ell^{-2})^{-\alpha} &\leq 4/3<2,
\end{split}
\end{equation}
By applying the inequalities in~(\ref{eq:caseI_intersection_kernel_step10a}) and~(\ref{eq:caseI_intersection_kernel_step10b}) to~(\ref{eq:caseI_intersection_kernel_step9}), we finally obtain that
\begin{align}\label{eq:caseI_intersection_kernel_step11}
\lambda_{\alpha,K}(z) \leq 4\textcolor{black}{c_5}\ell^{-1/2+2\alpha}\log \ell   +2\ell^{-\alpha\log \ell} 
\leq \textcolor{black}{16 c_5\ell^{-1/2}\log\ell +  \hspace{0.5ex}\ell^{-1}
= O(\ell^{-1/2}\log\ell),}
\end{align}
\textcolor{black}{which establishes the existence of the universal constant $c$ in~(\ref{eq:cases_middle}) and} concludes the proof.

\end{proof}

\begin{proof}[Proof of \eqref{eq:cases_tails}]
The proof of the tail intervals also follows from analyzing the average erasure probabilities. We present the proof mainly for the lower tail, where $z\in(0,1/\ell^2)$. Similar arguments yield the proof for the upper tail.

\textcolor{black}{Let $K\in GL(\ell,\Ftwo)$ denote an $\ell\times\ell$ random non-singular kernel. Define an indicator random variable $X_{i,s}$ as
\begin{align}
X_{i,s} = X_{i,s}(K) = \one(\mathcal{R}_{\ell-s} \cap (E_i\setminus E_{i-1}) = \emptyset),
\end{align}
where $\mathcal{R}_{\ell-s}$ is the linear span of the first $\ell-s$ columns in $K$ and $1\leq i \leq \ell$, and $0\leq s\leq \ell$. Given that $K$ is non-singular, $\mathcal{R}_{\ell-s}$ represents a random subspace of dimension $\ell-s$ in $\mathbb{F}_2^{\ell-s}$. By recalling the inequality in~(\ref{eq:randomkernel_caseI_erasure_probability_upperbound_step3}), we have
\begin{align}
\mathbb{E}[X_{i,s}] = p_{i|s} \leq  3 \left( \frac{2}{3}\right)^{i-s}.
\end{align}
To establish a concentration result for $X_{i,s}$, we use Markov's inequality to get
\begin{align}\label{eq:caseII_Markovs_Inequality}
\mathbb{P} \big( X_{i,s}\geq (1/3)\ell^2(\log \ell)(2/3)^{i-s}\big)\leq \frac{9}{\ell^2(\log \ell)}\hspace{1ex}.
\end{align}
Next, we apply the union bound over all values of $i$ and $s$ to get
\begin{align}\label{eq:caseII_union_bound}
\begin{split}
\mathbb{P} \big( X_{i,s}< (1/3)\ell^2(\log \ell)(2/3)^{i-s}, \hspace{1ex} \forall i,s\big) 
&= 1 - \mathbb{P} \big(\cup_{i,s} X_{i,s}\geq (1/3)\ell^2(\log \ell)(2/3)^{i-s}\big)\\
&\geq 1 - \ell(\ell+1).\frac{9}{\ell^2(\log \ell)} = 1-o(1).
\end{split}
\end{align}
Let us define the random variable $f_{K,i}(z)$ as
\begin{align}
f_{K,i}(z) = \sum_{s=1}^{\ell} X_{i,s}(K) \binom{\ell}{s}z^{s}(1-z)^{\ell-s}.
\end{align}
We can invoke the inequality in~(\ref{eq:caseII_union_bound}) to deduce that with probability at least $1-o(1)$ over the choice of $K$, we have
\begin{align}\label{eq:caseII_upperbound_f_i_z_step1}
\begin{split}
f_{K,i}(z) 
&\leq \sum_{s=1}^{\ell} (1/3)\ell^2(\log \ell)\bigg(\frac{2}{3}\bigg)^{i-s} \binom{\ell}{s}z^{s}(1-z)^{\ell-s}\\
&= (1/3)\ell^2(\log \ell)\bigg(\frac{2}{3}\bigg)^{i} \sum_{s=1}^{\ell} \binom{\ell}{s} \bigg(\frac{3z}{2}\bigg)^{s}(1-z)^{\ell-s}.
\end{split}
\end{align}  
}
Note that 
\begin{align}\label{eq:caseII_upperbound_f_i_z_step2}
\begin{split}
\sum_{s=1}^{\ell} \binom{\ell}{s} \bigg(\frac{3z}{2}\bigg)^{s}(1-z)^{\ell-s}& = \left(1+\frac{z}{2}\right)^{\ell} - (1-z)^\ell\\
&<(1+z)^{\ell} - (1-z)^{\ell} \\
&\leq 2\ell z(1+z)^{\ell-1},
\end{split}
\end{align}
where the last inequality in~(\ref{eq:caseII_upperbound_f_i_z_step2}) comes from the mean-value theorem for the function $f(x) = (1+x)^{\ell}$, which states that $f(x)-f(-x)= (x-(-x))f'(x_0)$ for some $x_0\in[-x,x]$. \textcolor{black}{Thus, for any $x\in(0,1)$, there exists $x_0$ in $(-x,x)$ such that $f(x)=2\ell x(1+x_0)^{\ell-1}\leq 2\ell x(1+x)^{\ell -1}$, since $(1+x)^{\ell-1}$ is increasing with $x$ for $x\geq 0$.}

\noindent Next, we point out that, for any $z< \ell^{-2}$ and any \textcolor{black}{$\ell\geq 32$}, we have 
\textcolor{black}{
\begin{align}\label{eq:caseII_upperbound_f_i_z_step4}
(1+z)^{\ell-1}\leq \bigg(1 + \frac{1}{\ell^2}\bigg)^{\ell-1}\leq \bigg(1 + \frac{1}{\ell^2}\bigg)^{\frac{\ell^2}{\ell-1}} < \exp(1/(\ell-1)) < 9/8. 
\end{align}
Now, we replace~(\ref{eq:caseII_upperbound_f_i_z_step2}) and~(\ref{eq:caseII_upperbound_f_i_z_step4}) in~(\ref{eq:caseII_upperbound_f_i_z_step1}) to deduce that with probability at least $1-o(1)$ over the choice of $K$, we have 
\begin{align}\label{eq:caseII_upperbound_f_i_z_step5}
f_{K,i}(z) < (3/4)\ell^3(\log \ell)\bigg(\frac{2}{3}\bigg)^i z.
\end{align}
For such kernels, we can use~(\ref{eq:caseII_upperbound_f_i_z_step5}) to derive the following upper bound on $\lambda_{\alpha,K}(z)$ for any $z\in(0,1/\ell^2)$:
\begin{align}\label{eq:caseII_upperbound_f_i_z_step6}
\begin{split}
\lambda_{\alpha,K}(z) 
&=\frac{\frac{1}{\ell}\sum_{i=1}^{\ell}g_{\alpha}(f_{K,i}(z))}{g_{\alpha}(z)}
=   \frac{1}{\ell} \sum_{i}^{\ell} \frac{\bigg(f_{K,i}(z)\big(1-f_{K,i}(z)\big) \bigg)^{\alpha} }{\big(z(1-z) \big)^{\alpha}}\hspace{5ex}\\
&< \frac{1}{\ell} \sum_{i=1}^{\ell} \bigg(f_{K,i}(z)\bigg)^{\alpha}z^{-\alpha}(1-z)^{-\alpha} 
< \ell^{3\alpha -1}(\log \ell)^{\alpha} \big(\frac{3/4}{1-z}\big)^{\alpha}\bigg(\sum_{i=1}^{\ell} \left(\frac{2}{3}\right)^{i\alpha} \bigg).\hspace{5ex}
\end{split}
\end{align}
Given that $\alpha > 0$ and $z<1/\ell^2< 1/4$, we have
\begin{equation}\label{eq:caseII_upperbound_f_i_z_step65}
\left(\frac{3/4}{1-z}\right)^{\alpha} < 1.
\end{equation}}
Furthermore,
\begin{align}\label{eq:caseII_upperbound_f_i_z_step7}
\begin{split}
\sum_{i=1}^{\ell} \left(\frac{2}{3}\right)^{i\alpha} 
< \sum_{i=1}^{\infty} \left(\frac{2}{3}\right)^{i\alpha} 
\textcolor{black}{\overset{(a)}{=} \frac{x}{1-x}\bigg|_{x=(\sfrac{2}{3})^{\alpha}}
= \frac{1}{x-1}\bigg|_{x=(\sfrac{3}{2})^{\alpha}}
\overset{(b)}{\leq} \frac{1}{\ln (3/2)^{\alpha}}
= \frac{1}{\alpha\ln(3/2)},}
\end{split}
\end{align}
where (a) comes from the power series expansion, and (b) is because of $\ln(x) \leq x-1$ for all $x>0$.
Moreover, 
\textcolor{black}{given that $\alpha = 1/\log \ell$,}
we obtain that
\begin{equation}\label{eq:caseII_upperbound_f_i_z_step8new}
\textcolor{black}{
\ell^{3\alpha} = 8 \text{ and } (\log \ell)^{\alpha} = (\log \ell)^{1/\log \ell} < 2.} 
\end{equation}
By combining \eqref{eq:caseII_upperbound_f_i_z_step6}, \eqref{eq:caseII_upperbound_f_i_z_step65}, \eqref{eq:caseII_upperbound_f_i_z_step7}, and \eqref{eq:caseII_upperbound_f_i_z_step8new}, 
we conclude that, for any $z\in(0,1/\ell^2)$, 
\begin{align}\label{eq:caseII_upperbound_f_i_z_step8}
\lambda_{\alpha,K}(z) < 
\textcolor{black}{
\frac{16}{\ln(3/2)}\ell^{-1}\log\ell = O(\ell^{-1}\log\ell),
}
\end{align}
which yields the desired bound on the lower tail. 
By following steps similar to~(\ref{eq:caseII_Markovs_Inequality})-(\ref{eq:caseII_upperbound_f_i_z_step8}), 
we can also show that, for any $z\in (1-1/\ell^2,1)$,  
\begin{align}\label{eq:caseII_upperbound_h_i_z_step4}
\lambda_{\alpha,K}(z) \leq c\ell^{-1}\log \ell,
\end{align}
\textcolor{black}{for some universal constant $c$ with probability at least $1-o(1)$} over the choice of the kernel. By combining \eqref{eq:caseII_upperbound_f_i_z_step8} and \eqref{eq:caseII_upperbound_h_i_z_step4} and using one last union bound, \textcolor{black}{we conclude that as $\ell$ grows, we have}
\begin{align}\label{eq:caseII_upperbound_h_i_z_step5}
\mathbb{P} \bigg\{\lambda_{\alpha,K}(z)< \textcolor{black}{c\ell^{-1}\log\ell}, \hspace{0.5em} \forall z\in\Big(0,\frac{1}{\ell^2}\Big)\cup\Big(1-\frac{1}{\ell^2},1\Big) \bigg\} > 1-\textcolor{black}{o(1)}.
\end{align}
\end{proof}

\section{Discussion and Open Problems} 
\label{sec:disc}

\textcolor{black}{This paper concerns the case of transmission over the binary erasure channel (BEC). One natural question is whether our results can be extended to the transmission over any binary memoryless symmetric channel (BMSC). After a preliminary version of this manuscript has appeared in \cite{pre1, pre2}, the question above has been resolved in \cite{Guru-Ria}. In the rest of this section,} we go over some unsolved challenges in the context of large-kernel polar codes. Most of these problems are initiated by the requirements on the size of the kernel. It was already mentioned that $\ell$ scales exponentially with the inverse of gap to the optimal scaling exponent, $\mu = 2$. This forces $\ell$ to be extremely large even for moderately good scaling exponents. In the following we address multiple issues with large $\ell$s:

\noindent$\hspace{1ex}\bullet\hspace{1ex}$\emph{Computation of the scaling exponent.} The computation of the scaling exponent even for the binary erasure channel is NP-hard~\cite{F14}. While there are methods to improve the efficiency of these calculations for small values of $\ell$, we are not aware of any algorithm that can do it for arbitrary $64\times 64$ kernels. 

\noindent$\hspace{1ex}\bullet\hspace{1ex}$\emph{Explicit construction of fast polarizing kernels.} In this paper, we showed that, given sufficiently large $\ell$, almost all binary non-singular $\ell\times\ell$ matrices are suitable polarization kernel candidates. However, the problem of finding one, or a family, of such kernels remains unsolved. Note that exhaustive search \textcolor{black}{only works up} to $\ell\sim8$, while there are a few heuristic construction algorithms for $\ell = 16,32,$ and $64$.    

\noindent$\hspace{1ex}\bullet\hspace{1ex}$\emph{Construction of polar codes.} In the polar coding terminology, the construction problem refers to the problem of \textcolor{black}{finding the best} bit-channels for which we transmit information over. There are multiple known algorithms for \textcolor{black}{classical} polar codes including the \textcolor{black}{Tal-Vardy} method in~\cite{TVbounds} \textcolor{black}{and the Gaussian Approximation in~\cite{Trif12}}. Unfortunately, there is yet another computation complexity blow-up if one replaces the $2\times 2$ kernels with arbitrarily large $\ell\times\ell$ matrices, which leaves us with the Monte-Carlo method for finding less noisy bit-channels. \textcolor{black}{However, this method is known to perform poorly in the precision/complexity trade-off.}  

\noindent$\hspace{1ex}\bullet\hspace{1ex}$\emph{Decoding complexity.} The recursive implementation of successive-cancellation decoding for polar codes is based on the butterfly-like graph, where each node represents a polarization kernel. These kernel-nodes perform successive-cancellation decoding of the kernel itself, and then communicate with each other  on a specific schedule to reveal the uncoded information bits sequentially and efficiently.
It is well known that the overall decoding complexity for 
conventional polar codes is $O(n\log n)$. However, the internal successive-cancellation computations within 
the kernels become more complicated when the $2\times2$ conventional kernel is replaced with an $\ell \times \ell$ kernel. This effectively changes the asymptotic decoding complexity to $O(2^{\ell} n \log n)$. This is probably the most controversial problem with using polar codes constructed from large kernels if the underlying channel \textcolor{black}{is not a BEC}. 
\textcolor{black}{Moreover, in the case of BEC, decoding can be accomplished by using Gaussian Elimination.}
This raises the main question about practicality of polar codes with large kernels. Recently, there have been multiple attempts at finding/constructing fast-polarizing kernels with enough structure that would allow us to design a decoding algorithm with reasonable decoding complexity. See for example~\cite{BFSTV17, TT19}. However, a general approach to reducing the decoding complexity of large kernels is still lacking from the literature.

Despite all these problems, we view the fact that polar codes constructed from random large kernels perform nearly as good as the random codes to be of theoretical interest. \textcolor{black}{That is, we have shown that polarization kernels with optimal scaling for the BEC exist. The problem of finding practical such kernels is a topic of further research.}


\bibliographystyle{acm}

\end{document}